\documentclass[11pt]{article}
\usepackage[left=2cm,right=2cm, top=2.5cm,bottom=2.5cm,bindingoffset=0cm]{geometry}
\usepackage{tikz, titlesec, tikz-cd}
\usetikzlibrary{calc,intersections,through,backgrounds}
\usepackage{amsmath, amssymb, amsthm}
\usepackage{mathtools}

\mathtoolsset{showonlyrefs}

\usepackage{array}

\usepackage[sort&compress, numbers]{natbib}

\usepackage{hyperref}

\usepackage{mathabx}

\titleformat{\section}{\bfseries\filcenter}{\thesection}{1em}{}
\titleformat{\subsection}{\bfseries}{\thesubsection}{1em}{}
\titleformat{\subsubsection}[runin]{\bfseries}{\thesubsubsection}{1em}{}[.]

\usepackage{multirow}

\usetikzlibrary{arrows} 
\tikzset{
    >=stealth',
    pil/.style={
           ->,
           thick,
           shorten <=2pt,
           shorten >=2pt,}
}

\newcommand{\Z}{\mathbb{Z}}
\newcommand{\C}{\mathbb{C}}
\newcommand{\R}{\mathbb{R}}
\renewcommand{\P}{\mathbb{P}}

\newcommand{\Id}{\mathrm{Id}}
\renewcommand{\Im}[1]{\mathrm{Im}\,#1}
\newcommand{\tr}{\mathrm{tr}}
\newcommand{\RP}{\mathbb{R}\mathbb{P}}
\newcommand{\GL}{\mathrm{GL}}
\newcommand{\D}{\mathcal D}

\newcommand{\sign}[1]{\mathrm{sgn}\,#1}
\newcommand{\sgn}[1]{\mathrm{sgn}\,#1}
\newcommand{\const}{\mathrm{const}}
\newcommand{\Ker}[1]{\mathrm{Ker}\,#1}
\newcommand{\ord}[1]{\mathrm{ord}\,#1}
\renewcommand{\v}[1]{Q^n({#1})}

\renewcommand{\r}{r}
\renewcommand{\l}{l}
\newcommand{\Dr}{\D_\r}
\newcommand{\Dl}{\D_\l}

\newtheorem{lemma}{Lemma}[section]
\newtheorem{proposition}[lemma]{Proposition}

\newtheorem{corollary}[lemma]{Corollary}

\newtheorem{innercustomthm}{Theorem}
\newenvironment{customthm}[1]
  {\renewcommand\theinnercustomthm{#1}\innercustomthm}
  {\endinnercustomthm}

\theoremstyle{definition}
\newtheorem{remark}[lemma]{Remark}
\newtheorem{definition}[lemma]{Definition}

\usepackage[nottoc]{tocbibind}

\title{The pentagram map, Poncelet polygons, and commuting difference operators}    
\author{Anton Izosimov\thanks{Department of Mathematics, University of Arizona, e-mail: \tt{izosimov@math.arizona.edu}}}
\date{}

\begin{document}
\maketitle
\abstract{

The pentagram map takes a planar polygon $P$  to a polygon $P'$ whose vertices are the intersection points of consecutive shortest diagonals of~$P$. This map is known to interact nicely with Poncelet polygons, i.e. polygons which are simultaneously inscribed in a conic and circumscribed about a conic. A theorem of R.\,Schwartz says that if $P$ is a Poncelet polygon, then the image of $P$ under the pentagram map is projectively equivalent to $P$. In the present paper we show that in the convex case this property characterizes Poncelet polygons: if a convex polygon is projectively equivalent to its pentagram image, then it is Poncelet. The proof is based on the theory of commuting difference operators, as well as on properties of real elliptic curves and theta functions.

}

\tableofcontents

\section{Introduction}
It is a classical result of A.\,Clebsch  \cite{clebschPaper} that every planar pentagon $P$ is projectively equivalent to the pentagon $P'$ formed by intersections of diagonals of~$P$.  More precisely, if we label the vertices of $P$ and $P'$ in such way that the $k$'th vertex of $P'$ is opposite to the $k$'th vertex of $P$ (see Figure \ref{Fig:pentagons}), then there is a projective transformation that takes $P$ to $P'$ and respects the labelings. 

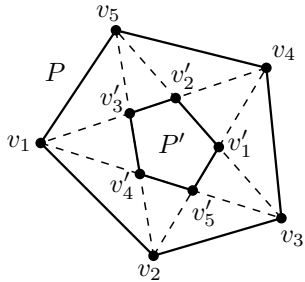
\begin{figure}[h]
\centering
\begin{tikzpicture}[ scale = 1]
\coordinate (A) at (0,1);
\coordinate (B) at (1.5,-0.5);
\coordinate (C) at (3.2,0);
\coordinate (D) at (3,2);
\coordinate (E) at (1,2.5);
\fill (A) circle [radius=2pt];
\fill (B) circle [radius=2pt];
\fill (C) circle [radius=2pt];
\fill (D) circle [radius=2pt];
\fill (E) circle [radius=2pt];

\node[label={[shift={(0,0)}]left:\small$P$}] at (0.6, 1.95) () {};
\node[label={[shift={(0,0)}]left:\small$P'$}] at (2.2, 1) () {};

\draw  [line width=0.3mm]  (A) -- (B) -- (C) -- (D) -- (E) -- cycle;
\draw [dashed, line width=0.2mm, name path=AC] (A) -- (C);
\draw [dashed,line width=0.2mm, name path=BD] (B) -- (D);
\draw [dashed,line width=0.2mm, name path=CE] (C) -- (E);
\draw [dashed,line width=0.2mm, name path=DA] (D) -- (A);
\draw [dashed,line width=0.2mm, name path=EB] (E) -- (B);

\node[label={[shift={(0.15,0)}]left:\small${v_1}$}] at (A) () {};
\node[label={[shift={(-0.05,0.15)}]below:\small$v_2$}] at (B) () {};
\node[label={[shift={(0.15,0.1)}]below:\small$v_3$}] at (C) () {};
\node[label={[shift={(-0.2,0.2)}]right:\small$v_4$}] at (D) () {};
\node[label={[shift={(-0.1,-0.15)}]above:\small$v_5$}] at (E) () {};

\path [name intersections={of=AC and BD,by=Ep}];
\path [name intersections={of=BD and CE,by=Ap}];
\path [name intersections={of=CE and DA,by=Bp}];
\path [name intersections={of=DA and EB,by=Cp}];
\path [name intersections={of=EB and AC,by=Dp}];

\fill (Ap) circle [radius=2pt];
\fill (Bp) circle [radius=2pt];
\fill (Cp) circle [radius=2pt];
\fill (Dp) circle [radius=2pt];
\fill (Ep) circle [radius=2pt];

\draw  [line width=0.3mm]  (Ap) -- (Bp) -- (Cp) -- (Dp) -- (Ep)  -- cycle;

\node[label={[shift={(-0.15,0.05)}]right:\small$v_1'$}]at (Ap) () {};
\node[label={[shift={(0.1,-0.15)}]above:\small$v_2'$}]at (Bp) () {};
\node[label={[shift={(0.2,0.2)}]left:\small$v_3'$}] at (Cp) () {};
\node[label={[shift={(0.2,-0.15)}]left:\small$v_4'$}] at (Dp) () {};
\node[label={[shift={(0.15,0.15)}]below:\small$v_5'$}] at (Ep) () {};

\end{tikzpicture}
\caption{Pentagons $P = v_1v_2v_3v_4v_5$ and $P' = v_1'v_2'v_3'v_4'v_5'$ are projectively equivalent.}\label{Fig:pentagons}
\end{figure}


Furthermore, as was proved by R.\,Schwartz \cite{schwartz2007poncelet}, Clebsch's theorem is true for all  {Poncelet} polygons with odd number of vertices. Recall that a \textit{Poncelet polygon} is a polygon which is inscribed in a conic and circumscribed about another conic. In particular, any pentagon is Poncelet, while for $n$-gons with $n\geq 6$ being Poncelet is a non-trivial restriction. Poncelet polygons owe their name to J.-V.\,Poncelet and his famous ``{porism}'' which says that {if there exists an $n$-gon inscribed in a conic $C_1$ and circumscribed about a conic $C_2$, then any point of $C_1$ is a vertex of such an $n$-gon} (see Figure \ref{Fig:poncelet}). 
\begin{figure}[t]
\centering
\includegraphics[width = 5 cm]{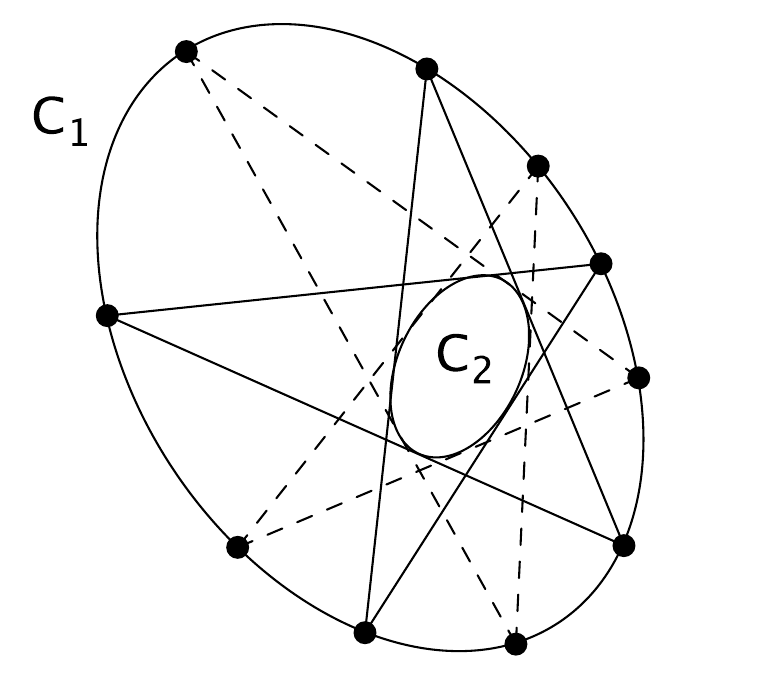}
\caption{Every point of $C_1$ is a vertex of a pentagon inscribed in $C_1$ and circumscribed about $C_2$.}\label{Fig:poncelet}
\end{figure}

\begin{figure}[b]
\centering
\begin{tikzpicture}[scale = 1.1]
\coordinate (A) at (0,0);
\coordinate (B) at (1.5,-0.5);
\coordinate (C) at (3,1);
\coordinate (D) at (3,2);
\coordinate (E) at (1,3);
\coordinate (F) at (-0.5,2.5);
\coordinate (G) at (-1,1.5);

\fill (A) circle [radius=2pt];
\fill (B) circle [radius=2pt];
\fill (C) circle [radius=2pt];
\fill (D) circle [radius=2pt];
\fill (E) circle [radius=2pt];
\fill (F) circle [radius=2pt];
\fill (G) circle [radius=2pt];

\draw  [line width=0.3mm]  (A) -- (B) -- (C) -- (D) -- (E) -- (F) -- (G) -- cycle;
\draw [dashed, line width=0.2mm, name path=AC] (A) -- (C);
\draw [dashed,line width=0.2mm, name path=BD] (B) -- (D);
\draw [dashed,line width=0.2mm, name path=CE] (C) -- (E);
\draw [dashed,line width=0.2mm, name path=DF] (D) -- (F);
\draw [dashed,line width=0.2mm, name path=EG] (E) -- (G);
\draw [dashed,line width=0.2mm, name path=FA] (F) -- (A);
\draw [dashed,line width=0.2mm, name path=GB] (G) -- (B);

\path [name intersections={of=AC and BD,by=Fp}];
\path [name intersections={of=BD and CE,by=Gp}];
\path [name intersections={of=CE and DF,by=Ap}];
\path [name intersections={of=DF and EG,by=Bp}];
\path [name intersections={of=EG and FA,by=Cp}];
\path [name intersections={of=FA and GB,by=Dp}];
\path [name intersections={of=GB and AC,by=Ep}];

\fill (Ap) circle [radius=2pt];
\fill (Bp) circle [radius=2pt];
\fill (Cp) circle [radius=2pt];
\fill (Dp) circle [radius=2pt];
\fill (Ep) circle [radius=2pt];
\fill (Fp) circle [radius=2pt];
\fill (Gp) circle [radius=2pt];

\draw  [line width=0.3mm]  (Ap) -- (Bp) -- (Cp) -- (Dp) -- (Ep) -- (Fp) -- (Gp) -- cycle;

\node[label={[shift={(0.2,-0.1)}]left:\small$v_1$}] at (A) () {};
\node[label={[shift={(0.1,0.15)}]below:\small$v_2$}] at (B) () {};
\node[label={[shift={(0.25,0.25)}]below:\small$v_3$}] at (C) () {};
\node[label={[shift={(-0.2,0.15)}]right:\small$v_4$}] at (D) () {};
\node[label={[shift={(0,-0.15)}]above:\small$v_5$}] at (E) () {};
\node[label={[shift={(0.15,0.1)}]left:\small$v_6$}] at (F) () {};
\node[label={[shift={(0.15,0)}]left:\small$v_7$}] at (G) () {};

\node[label={[shift={(0.25,-0.25)}]left:\small$v_1'$}] at (Ap) () {};
\node[label={[shift={(0.1,0.15)}]below:\small$v_2'$}] at (Bp) () {};
\node[label={[shift={(0.25,0.3)}]below:\small$v_3'$}] at (Cp) () {};
\node[label={[shift={(-0.2,0.15)}]right:\small$v_4'$}] at (Dp) () {};
\node[label={[shift={(0,-0.15)}]above:\small$v_5'$}] at (Ep) () {};
\node[label={[shift={(0.2,0.2)}]left:\small$v_6'$}] at (Fp) () {};
\node[label={[shift={(0.1,0)}]left:\small$v_7'$}] at (Gp) () {};
\node at (1,1.5) () {$P'$};
\node at (0,3) () {$P$};

\end{tikzpicture}
\caption{A convex polygon $P$ is Poncelet if and only if it is projectively equivalent to $P'$.}\label{Fig:labeling}
\end{figure}
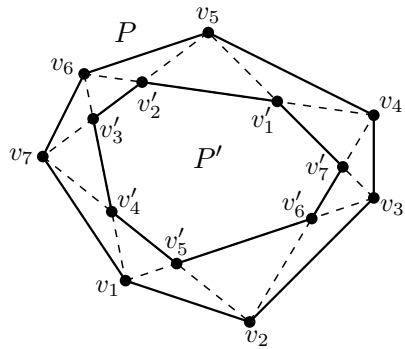
Schwartz's generalization of Clebsch's theorem is as follows. Let $P$ be an $n$-gon with odd $n$, and let $P'$ be the polygon whose vertices are the intersections of consecutive shortest diagonals of~$P$, i.e. diagonals connecting second nearest vertices. Label the vertices of $P'$ as in Clebsch's theorem: the $k$'th vertex of $P'$ is opposite to the $k$'th vertex of $P$ (see Figure~\ref{Fig:labeling}). Assume that $P$ is Poncelet. Then there is a projective transformation that carries $P$ to $P'$ and respects the labelings (a weaker result saying that if $P$ is Poncelet then $P'$ is circumscribed about a conic was known to Darboux, see \cite[Theorem~2.1]{dragovic2014bicentennial}). 
The goal of the present paper is to show that in the convex setting the converse is also true. More precisely, we prove the converse of Schwartz's theorem for a broader class of \textit{weakly convex} polygons. Weak convexity is a technical condition (see Definition~\ref{def:lcp} below) which in particular holds for truly convex polygons. 

\begin{customthm}{A}\label{thm1}
Let $P$ be a weakly convex closed polygon with an odd number of vertices. Let also $P'$ be the polygon whose vertices are the intersections of consecutive shortest diagonals of~$P$, labeled as in Figure~\ref{Fig:labeling}. Assume that there is a projective transformation that carries $P$ to $P'$ and respects the labelings. Then $P$ (and hence $P'$) is a Poncelet polygon.
\end{customthm}

Combining Theorem \ref{thm1} with Schwartz's theorem, we get the following characterization of weakly convex Poncelet polygons:

\begin{corollary}\label{thm1cor}
Let $P$ be a weakly convex closed polygon with an odd number of vertices. Let also $P'$ be the polygon whose vertices are the intersections of consecutive shortest diagonals of~$P$, labeled as in Figure~\ref{Fig:labeling}. Then $P$ is Poncelet if and only if it is projectively equivalent to $P'$.
\end{corollary}

The map taking the polygon $P$ in Figure \ref{Fig:labeling} to $P'$ is known as the \textit{pentagram map}. It was defined by Schwartz in 1992 \cite{schwartz1992pentagram} but became especially popular in the last decade thanks to the discovery that it is a discrete integrable system \cite{ovsienko2010pentagram, ovsienko2013liouville, soloviev2013integrability}, and also because of its connections with cluster algebras~\cite{GLICK20111019, Gekhtman2016, glick2015, kedem2015t, fock2014loop}. 
Since the pentagram map commutes with projective transformations, it is usually considered as a dynamical system on the space of polygons modulo projective equivalence. Our result can thus be viewed as a description of fixed points of the pentagram map, which has been an open question since Schwartz's first paper \cite{schwartz1992pentagram}.
\begin{remark}
Note that the pentagram map can be considered either on labeled polygons (i.e. polygons with labeled vertices), or on unlabeled ones.  Theorem~\ref{thm1} describes fixed points of the pentagram map on the space of projective equivalence classes of \textit{labeled} polygons, where the labeling rule is depicted in Figure~\ref{Fig:labeling}. Although this is not the only possible labeling, it is the only one for which the pentagram map commutes with the action of the dihedral group, and hence the most symmetric of all labelings. A more common, non-symmetric labeling is given by the rule  $v_k' := (v_{k-1}, v_{k+1}) \cap (v_k, v_{k+2})$. One can easily see that the only fixed points of the pentagram map with this non-symmetric labeling are regular polygons (again, assuming that the number of vertices is odd). 
The problem  of describing the fixed points of the pentagram map for an \textit{arbitrary} labeling can also be approached using the techniques of the present paper, but due to the break of symmetry one should not expect an answer as nice as for the symmetric labeling.
\end{remark}

Theorem~\ref{thm1} also has an interpretation in terms of billiards. Indeed, if the conic $C_1$ circumscribed about a Poncelet polygon $P$ is confocal to the inscribed conic $C_2$ (which can always be arranged by applying a suitable projective transformation), then $P$ can be viewed as a closed trajectory of a billiard ball in the domain bounded by $C_1$. Conversely, any closed billiard trajectory in a conic is a Poncelet polygon.
So, Corollary \ref{thm1cor} establishes a correspondence between fixed points of the pentagram map and periodic billiard trajectories in conics. Also note that, as shown in~\cite{levi2007poncelet}, the fact that a closed billiard trajectory in a conic is projectively equivalent to its pentagram image is essentially a corollary of integrability of the corresponding billiard system. At the same time, we show that if $P$ is projectively equivalent to its pentagram image then the vertices of $P$ are contained in a conic. So, one may hope to combine our results with the approach of \cite{levi2007poncelet} to show that for \textit{any} integrable billiard the impact points of a periodic trajectory are contained in a conic, and hence shed some light on the Birkhoff conjecture which says that the only integrable billiards are the ones in conics.

It is also an interesting question whether this correspondence between periodic billiard trajectories in conics and fixed points of the pentagram map extends to higher dimensions. There exists numerous generalizations of the pentagram map to higher-dimensional spaces~\cite{Gekhtman2016, khesin2013, khesin2016, felipe2015} and one may wonder if their fixed points are related to periodic trajectories of billiards in multidimensional quadrics. 




\begin{remark}\label{ex:complex}
We do not know if Theorem \ref{thm1} is true with no convexity-type assumptions, but it is for sure not true over complex numbers, as demonstrated by the following example.
Let $\lambda := \exp({{2\pi\mathrm{i}}/{7}})$, where $\mathrm{i} = \sqrt{-1}$, and let $P$ be a heptagon in $\C^2$ with vertices $v_k := (\lambda^{2k}, \lambda^{3k})$. Then a direct computation shows that there exists a {projective} (in fact, even affine) transformation $\phi$ taking $P$ to $P'$ (see also Remark \ref{rem:crit} below for a conceptual proof). Moreover, for any vertex $v$ of $P$ and $v'$ of $P'$, the map $\phi$ can be chosen to take $v$ to $v'$. This means that the projective equivalence class of $P$ is fixed by the pentagram map, regardless of the labeling convention used to define the map. However, $P$ is not Poncelet. Moreover, it is not even inscribed in a conic. Indeed, the vertices of $P$ lie on a semi-cubical parabola $y^2 = x^3$, which has at most six intersection points with any conic. So, there exists no conic which contains all seven vertices of $P$.\par
Theorem \ref{thm1} being not true over $\C$ is one of the reasons one should not expect it to have any kind of ``elementary'' proof, as such a proof would be valid over any field. Another reason is that the theorem is not true for non-closed polygons (see Remark \ref{rm:nc} below). Again, if there was a local (i.e. involving only few adjacent vertices) geometric construction producing inscribed and circumscribed conics for $P$ based on the projective equivalence between $P$ and $P'$, such a construction would work no matter whether $P$ is closed or not.
\end{remark}

We now outline the scheme of the proof of Theorem \ref{thm1}. As a first step, we prove the theorem under an additional assumption that the polygon $P$ is \textit{self-dual}. 
Recall that the dual of a polygon is the polygon in the dual projective plane whose vertices are the sides of the initial one. We label the vertices of the dual polygon as shown in Figure \ref{Fig:dual}. A polygon is self-dual if it is projectively equivalent to its dual.
\begin{figure}[t]
\centering
\begin{tikzpicture}[scale = 0.9]
\coordinate (A) at (0,0);
\coordinate (B) at (1.5,-0.5);
\coordinate (C) at (3,1);
\coordinate (D) at (3,2);
\coordinate (E) at (1,3);
\coordinate (F) at (-0.5,2.5);
\coordinate (G) at (-1,1.5);

\coordinate (Ap) at (2,2.5);
\coordinate (Bp) at (0.25,2.75);
\coordinate (Cp) at (-0.75,2);
\coordinate (Dp) at (-0.5,0.75);
\coordinate (Ep) at (0.75,-0.25);
\coordinate (Fp) at (2.25,0.25);
\coordinate (Gp) at (3,1.5);

\draw  [line width=0.3mm]  (A) -- (B) -- (C) -- (D) -- (E) -- (F) -- (G) -- cycle;

\node[label={[shift={(0.2,-0.1)}]left:\small$v_1$}] at (A) () {};
\node[label={[shift={(0.1,0.15)}]below:\small$v_2$}] at (B) () {};
\node[label={[shift={(0.25,0.25)}]below:\small$v_3$}] at (C) () {};
\node[label={[shift={(-0.2,0.15)}]right:\small$v_4$}] at (D) () {};
\node[label={[shift={(0,-0.15)}]above:\small$v_5$}] at (E) () {};
\node[label={[shift={(0.15,0.1)}]left:\small$v_6$}] at (F) () {};
\node[label={[shift={(0.15,0)}]left:\small$v_7$}] at (G) () {};

\node[label={[shift={(0.25,-0.2)}]left:\small$v_1^*$}] at (Ap) () {};
\node[label={[shift={(0.1,0.15)}]below:\small$v_2^*$}] at (Bp) () {};
\node[label={[shift={(0.25,0.3)}]below:\small$v_3^*$}] at (Cp) () {};
\node[label={[shift={(-0.2,0.15)}]right:\small$v_4^*$}] at (Dp) () {};
\node[label={[shift={(0,-0.15)}]above:\small$v_5^*$}] at (Ep) () {};
\node[label={[shift={(0.2,0.2)}]left:\small$v_6^*$}] at (Fp) () {};
\node[label={[shift={(0.1,0)}]left:\small$v_7^*$}] at (Gp) () {};
\fill (A) circle [radius=2pt];
\fill (B) circle [radius=2pt];
\fill (C) circle [radius=2pt];
\fill (D) circle [radius=2pt];
\fill (E) circle [radius=2pt];
\fill (F) circle [radius=2pt];
\fill (G) circle [radius=2pt];

\end{tikzpicture}
\caption{The $k$'th vertex of the dual polygon is opposite to the $k$'th vertex of the initial one.}\label{Fig:dual}
\end{figure}
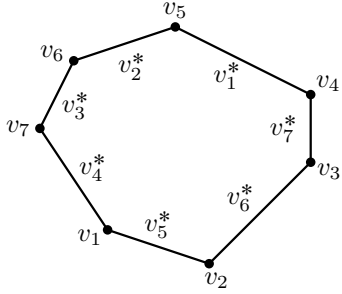

Under this additional assumption of self-duality, Theorem \ref{thm1} is true in a more general setting of \textit{twisted polygons}, i.e. polygons that are closed only up to a projective transformation. More precisely, a twisted $n$-gon is a sequence $v_k \in \P^2$ such that $v_{k+n} = \psi(v_k)$ for a certain projective transformation~$\psi$, called the \textit{monodromy}.

\begin{customthm}{B}\label{thm2}
Let $P$ be a weakly convex twisted $n$-gon with odd $n$, and let $P'$ be as in Theorem \ref{thm1}. Assume that $P$ is self-dual and projectively equivalent to $P'$. Then $P$ is a Poncelet polygon.
\end{customthm}

\begin{remark}\label{rm:nc}
Theorem \ref{thm2} is not true without the self-duality assumption (in other words, Theorem~\ref{thm1} is not true for non-closed polygons). As an example, consider a polygon $P$ in $\R^2$ whose vertices are given by $v_k := (4^k, 8^k)$. This is a twisted $n$-gon for any $n$, weakly convex and projectively equivalent to its pentagram image $P'$. %
%
%
%
However, $P$ is not Poncelet (cf. Remark \ref{ex:complex}).
\end{remark}

The proof of Theorem \ref{thm2} is based on the theory of commuting difference operators, elliptic curves, and theta functions. Given the result of the theorem, the appearance of elliptic curves is not surprising, as their connection to Poncelet polygons is well-known. However, we are not given \textit{a priori} that $P$ is Poncelet, so there should be some other source where the elliptic curve is coming from. In our approach, that source is the theory of commuting difference operators. Namely, we show that a twisted polygon $P$ is projectively equivalent to its pentagram image $P'$ if and only if certain associated difference operators commute (see Section \ref{sec:fpcdo}). The general theory then says that the joint spectrum of those operators is a Riemann surface $\Gamma$, called the \textit{spectral curve} (see e.g.~\cite{krichever2003two}). Using further that the operators in question are of special form and, in particular, dual to each other (which is a reflection of self-duality of $P$), we show that the genus of $\Gamma$ is at most $1$, i.e. $\Gamma$ is rational or elliptic (see Section~\ref{ss:genus}). This is one of the few places in the proof where we use weak convexity in an essential way. Without that assumption, one only seems to be able to conclude that the genus is at most $2$. In particular, it seems possible to construct non-weakly-convex counterexamples to Theorems \ref{thm1} and \ref{thm2} using genus $2$ theta functions.

\par

The next step is to show that our upper bound on the genus of $\Gamma$ implies that $P$ is Poncelet (as a priori there is no connection between the elliptic curve $\Gamma$ and the one that is classically associated with a Poncelet polygon). More precisely, we show that elliptic spectral curves correspond to generic Poncelet polygons, while rational curves correspond to their degenerations (such as the regular polygon). To that end, we express coordinates of vertices of $P$ in terms of certain meromorphic functions on the spectral curve $\Gamma$. One ends up with elementary functions or theta functions, depending on whether the curve $\Gamma$ is rational or elliptic. In both cases, using relations between those functions (e.g. Riemann's relation in the elliptic case), one shows that $P$ is a Poncelet polygon, so Theorem \ref{thm2} holds (see Section~\ref{sec:rat} for the rational case and Section \ref{ss:sd} for the elliptic case). Furthermore, in the elliptic case the spectral curve $\Gamma$ turns out to be isogenous to the elliptic curve attached to $P$ due to its Poncelet property.

\par

After that, we proceed to prove Theorem \ref{thm1}. To that end, we first show that the self-duality assumption of Theorem \ref{thm2} is not too restrictive. Namely, given a polygon that is projectively equivalent to its pentagram image, it can be what is called \textit{rescaled} so that it becomes self-dual. This rescaling (closely related to the notion of \textit{spectral parameter}) is a one-parametric group of transformations of the moduli space of twisted polygons which preserves weak convexity but not closedness. So, starting with a closed polygon as in Theorem \ref{thm1}, we rescale it to a non-necessarily closed, but self-dual, weakly convex polygon, which is the setting of Theorem \ref{thm2}. This way, we conclude that a weakly convex closed polygon projectively equivalent to its pentagram image is Poncelet up to rescaling. The last step is to show that this rescaling must actually be trivial. To that end, we show that no non-trivial rescaling of a weakly convex Poncelet polygon is closed. In the rational case, this is proved by an elementary argument (see Section \ref{ss:genrat}), while the elliptic case requires careful analysis of the real part of the spectral curve (see Section~\ref{ss:genell}). In the latter case, the proof once again essentially relies on weak convexity. 

In addition to proofs of Theorems \ref{thm1} and \ref{thm2}, the paper contains an appendix (Section \ref{sec:app}) where we establish an auxiliary result on correspondence between dual difference operators and dual polygons. Although that result seems to be well-known, we could not find a proof in the literature that does not rely on a computation. So, we provide a proof here.\par

We tried to make the exposition self-contained. In particular, we do not assume
that the reader is familiar with the general theory of integrable systems or commuting difference operators. Only basic knowledge of Riemann surfaces is assumed.

\par 

\smallskip
{\bf Acknowledgments.} The author is grateful to Boris Khesin, Valentin Ovsienko, Richard Schwartz, and Sergei Tabachnikov for fruitful discussions. Some of the figures were created with help of software package Cinderella.  
 This work was supported by NSF grant DMS-2008021.

\section{Background results: polygons, difference operators, and corner invariants}

This section is an overview of mostly well known results on the relation between difference operators and polygons. Namely, we give an introduction to difference operators in Section \ref{sec:primer}, after which we connect them to polygons in Section \ref{sec:polygons}. Note that while the description of the moduli space of polygons in terms of difference operators is well-known, our point of view is slightly different from the standard one. In particular, we identify the space of polygons with a certain \textit{quotient} of the space of third order difference operators, as opposed to the standard approach in which one identifies polygons with a certain \textit{subspace} of that space. In that respect, our approach is close to that of \cite{conley2019lagrangian}. In addition, we provide, in Section \ref{sec:ci}, another description of the polygon space, in terms of so-called corner invariants. Note that while corner invariants per se are not heavily used in the paper, they are needed to define weakly convex polygons and rescaling. Rescaling is also defined in Section \ref{sec:ci}, while weakly convex polygons are discussed in the next Section \ref{sec:lcp}. 

\par

\subsection{A primer on difference operators}\label{sec:primer}

This section is a brief introduction to the elementary theory of difference operators. Our terminology mainly follows that of \cite{van1979spectrum}.
Let $\R^\infty$ be the vector space of bi-infinite sequences of real numbers. For $\xi \in \R^\infty$ and any $k \in \Z$, let $\xi_k \in \R$ be the $k$'th entry of the sequence $\xi$, so that $\xi = (\xi_k)_{k \in \Z}$. Let also $m_- \leq m_+$ be integers. A linear operator $\D \colon \R^\infty \to \R^\infty$ is called \textit{a difference operator supported in $[m_-,m_+]:=\{m_-, \dots, m_+\}$} if it can be written as
\begin{equation}\label{dodef}
(\D\xi)_k = \! \sum_{j = m_-}^{m_+} \! a_k^j \xi_{k+j},
\end{equation}
where $a_k^j \in \R$ for every $k \in \Z$ and every $j \in [m_-,m_+]$. In matrix terms, this can be rewritten as
	\begin{align}\label{infMatrix0}
	\D\xi = \left(\begin{array}{ccccccccc}  \ddots   & & \ddots & & \\  & a_{k-1}^{m_-}   & \dots & a_{k-1}^{m_+} & \\ & & a_k^{m_-} &  \dots & a_k^{m_+} & &  \\ & & & a_{k+1}^{m_-} & \dots & a_{k+1}^{m_+} &\\ & & &  &  \ddots & &\ddots \end{array}\right)\xi,
	\end{align}
so difference operators can be equivalently described as those whose matrices are {finite band} (i.e. have only finitely many non-zero diagonals). Furthermore, denoting, for every $j$, the sequence of $a_k^j$'s by $a^j$, formula \eqref{dodef} can be rewritten as
	\begin{align}\label{genDiffOp}
\D = \!\sum_{j = m_-}^{m_+} a^j T^j,
\end{align}
where $T \colon \R^\infty \to \R^\infty$ is the left shift operator $(T\xi)_k = \xi_{k+1}$, and each $a^j \in \R^\infty$ acts on $\R^\infty$ by term-wise multiplication. 
\par
The \textit{order} of difference operator \eqref{dodef} is the number $\ord \D := m_+ - m_-$. Difference operator~\eqref{dodef} is called \textit{properly bounded} if $a_k^{m_-} \!\neq 0$ and $a_k^{m_+} \!\neq 0$ for every $k \in \Z$. Clearly, for a properly bounded difference operator $\D$ one has
$
\dim \Ker \D = \ord \D.
$ Sequences $\xi^{1}, \dots, \xi^{d} \in \Ker \D$, where $d:=\ord \D$, form a basis in $\Ker \D$ if and only if the associated \textit{difference Wronskian}
$$
W_k :=  \left|\begin{array}{ccc}\xi^{1}_{k} & \dots & \xi^{d}_{k} \\ & \dots   \\  \xi^{1}_{k+d-1}  & \dots & \xi^{d}_{k+d-1} \end{array}\right|,
$$
where $|M|$ stands for the determinant of the matrix $M$, is non-vanishing for some $k \in \Z$. This is equivalent to non-vanishing of $W_k$ for any $k$ due to the relation
\begin{equation}\label{wrons}
  W_{k+1} = (-1)^d \frac{a_{k-m_-}^{m_-}}{a_{k-m_-}^{m_+}}  W_k.
\end{equation}
\par
Along with $\R^\infty$, difference operators naturally act on the space $(\R^d)^\infty$ of bi-infinite sequences of vectors in $\R^d$. The case $d = \ord \D$ is of particular interest. Let $V  \in (\R^d)^\infty$, where $d = \ord \D$, be a solution of the difference equation $\D V = 0$. Define scalar sequences $\xi^{1}, \dots, \xi^{d} \in \R^\infty$ by setting $\xi_k^j$ to be equal to the $j$'th coordinate of $V_k$. We say that $V$ is a \textit{fundamental solution} if the sequences $\xi^{1}, \dots, \xi^{d} \in \R^\infty$ form a basis in $\Ker \D$. As follows from the Wronskian criterion, a solution of $\D V = 0$ is fundamental if and only if the vectors $V_k, \dots, V_{k + d -1}$ are linearly independent for some (equivalently, for all) $k \in \Z$.

A difference operator $\D$ is \textit{$n$-periodic} if its coefficients $a_k^j$ are $n$-periodic in the index $k$. This is equivalent to saying that $\D$ commutes with the $n$'th power of the shift operator: $\D T^n = T^n\D$, so the kernel of an $n$-periodic operator $\D$ is invariant under the action of $T^n$. The finite-dimensional operator $T^n\vert_{\Ker \D}$ is called the \textit{monodromy} of $\D$. Note that the eigenvectors of the monodromy are exactly scalar {quasi-periodic} solutions of the equation $\D\xi = 0$, i.e. solutions which belong to the space
\begin{equation}\label{qspace}
\v{z} := \{ \xi \in \R^\infty \mid \xi_{k+n} = z\xi_k\}
\end{equation} 
for some $z \in \R^*$. \par

The monodromy can also be understood in terms of fundamental solutions. 
Namely, notice that any two fundamental solutions $V, V' \in  (\R^d)^\infty$ of $\D$ are related by $V' = AV$, where $A \in \GL_d(\R)$ acts on $(\R^d)^\infty$ by term-wise multiplication. Furthermore, if $V$ is a fundamental solution of an $n$-periodic operator $\D$, then so is $T^nV$, which means that $T^n V =  AV$ for some $A \in \GL_d(\R)$. In other words, we have $V_{k+n} = AV_k$ for every $k \in \Z$, which means that the fundamental solution of a periodic difference operator is always quasi-periodic. Furthermore, the matrix $A$ can be easily seen to be the transpose of the monodromy matrix $M$ of $\D$, written in the basis of $\Ker \D$ associated with the fundamental solution~$V$.  In particular, this implies that the Wronskian of an $n$-periodic operator $\D$ satisfies $W_{k+n} = (\det M)W_k$. Combined with
 \eqref{wrons}, the latter formula gives the following expression for the determinant of the monodromy, which will be used several times throughout the paper:
\begin{equation}\label{monodet}
\det M = (-1)^{nd} \prod_{k=1}^n \frac{ a_k^{m_-}}{a_k^{m_+}}.
\end{equation}


The \textit{dual} of the operator~\eqref{genDiffOp} is defined by
 $$
 \D^* := \sum_{j = l}^{m} T^{-j} a^j = \sum_{j = l}^{m} \tilde a^jT^{-j},
 $$
 where $\tilde a^j_k = a_{k-j}^j$.
In other words, 
$\D^*$ is the formal adjoint of $\D$ with respect to the $l^2$ inner product on $\R^\infty$, i.e.
$
\langle \xi, \D \eta \rangle = \langle \D^*\xi,  \eta \rangle
$
whenever at least one of these inner products is well-defined. In the periodic case, the duality between $\D$ and $\D^*$ can also be understood as follows. If $\D$ is an $n$-periodic operator, then $\D^*$ is $n$-periodic as well. Furthermore, the formula
\begin{equation}\label{pairing}
\langle \xi, \eta \rangle := \sum_{k=1}^n \xi_k\eta_k
\end{equation}
defines an inner product on the space $\v{1}$ of $n$-periodic sequences, and the restrictions of $\D$ and $\D^*$ to  $\v{1}$ are dual to each other  with respect to this inner product. More generally, for every $z \in \R^*$, the restriction of  $\D$ to $\v{z}$ is dual to the restriction of $\D^*$ to $\v{z^{-1}}$ with respect to the pairing between $\v{z}$ and $\v{z^{-1}}$ given by the same formula \eqref{pairing}. As a corollary, we have
$$
\dim \Ker \D^*\vert_{\v{z^{-1}}} = \dim \Ker \D\vert_{\v{z}}.
$$
In particular, a non-zero number $z \in \R^*$ is an eigenvalue of the monodromy of $\D$ if and only if $z^{-1}$ is an eigenvalue of the monodromy of $\D^*$. 

\par 

\subsection{Difference operators and polygons}\label{sec:polygons}
In this section we describe the space of projective equivalence classes of planar polygons as a certain quotient of third order difference operators.

\begin{definition}
A \textit{polygon} in $\RP^2$ is a bi-infinite sequence of points $v_k \in \RP^2$ satisfying the following \textit{$3$-in-a-row} condition: for every $k \in \Z$ the points $v_{k-1}, v_k, v_{k+1}$ are in general position.
\end{definition}
Polygons modulo projective transformations can be encoded by means of properly bounded third order difference operators, i.e. operators of the form
\begin{equation}\label{o3op}
\D = aT^j+ bT^{j+1} + cT^{j+2} + dT^{j+3},
\end{equation}
where $a,b,c,d \in \R^\infty$ are such that $a_k \neq 0$, $d_k \neq 0$ for any $k \in \Z$. 
\begin{proposition}[cf. Proposition 4.1 of \cite{ovsienko2010pentagram}]\label{freedom}
For any $j \in \Z$, there is a one-to-one correspondence between projective equivalence classes of planar polygons and properly bounded difference operators $\D$ supported in $[j,j+3]$, considered up to the action $\D \mapsto \lambda\circ\D\circ \mu^{-1}$, where $\lambda, \mu \in (\R^*)^\infty$ are sequences of non-zero real numbers, acting on $\R^\infty$ by term-wise multiplication.

\end{proposition}

\begin{proof}
Given a properly bounded difference operator $\D$ supported in $[j,j+3]$, consider its fundamental solution $V$, which is a sequence of non-zero vectors in $\R^3$  (see Section~\ref{sec:primer}). Each term $V_k$ of that sequence determines a point $v_k \in \RP^2$ with homogeneous coordinates given by $V_k$. Furthermore, since $V$ is a fundamental solution, the vectors $V_k$, $V_{k+1}$, $V_{k+2}$ are linearly independent, and thus the sequence  $v_k \in \RP^2$ satisfies the {$3$-in-a-row} condition. Notice also that since the fundamental solution $V$ is unique up to a linear transformation $V \mapsto AV$, it follows that the polygon $\{v_k\}$ is well-defined up to projective equivalence. Thus, with each properly bounded difference operator $\D$ supported in $[j,j+3]$ one can associate a polygon $\{v_k\}$, defined up to a projective transformation. Conversely, given a polygon $\{v_k\}$, one can revert this construction to obtain a properly bounded difference operator $\D$ supported in $[j,j+3]$. To that end, one first lifts every point $v_k \in \P^2$ to a vector $V_k \in \R^3$, and then finds an operator $\D$ whose fundamental solution is given by $V$. Since the lifts $V_k$ of the points $v_k$ are unique up to a transformation $V_k \mapsto \mu_k V_k$, while the choice of an operator $\D$ with a given fundamental solution $V$ is unique up to $\D \mapsto \lambda \circ \D$, where $\lambda\in (\R^*)^\infty$, it follows that the operator $\D$ corresponding to a given polygon is defined up to the action  $\D \mapsto \lambda \circ \D \circ \mu^{-1}$, as desired.
\end{proof}

In what follows, we will be interested in closed and, more generally, twisted polygons. A closed $n$-gon is a polygon satisfying $v_{k+n} = v_k$ for every $k \in \Z$. For such a polygon, the corresponding difference operator $\D$ can be chosen to be $n$-periodic. The converse is, however, not true: polygons corresponding to periodic operators are, in general, not closed but twisted. Indeed, if $\D$ is a periodic operator, then its fundamental solution $V \in (\R^3)^\infty$ is in general not periodic but satisfies $ V_{k+n} = AV_k$, where $A \in \mathrm{GL}_3(\R)$ is the transposed monodromy of~$\D$. Therefore, the corresponding polygon satisfies $v_{k+n} = \psi(v_k)$, where $\psi \in \P \mathrm{GL}_3(\R)$ is the projective transformation determined by the linear operator~$A$. 
\begin{definition}
A \textit{twisted $n$-gon} in $\RP^2$ is a polygon $\{v_k\}$ which satisfies $v_{k+n} = \psi(v_k)$ for some projective transformation $\psi \in \P \mathrm{GL}_3$, called the \textit{monodromy of the polygon}.
\end{definition}
The above construction (see the proof of Proposition \ref{freedom}) allows one to identify the space of projective equivalence classes of twisted $n$-gons with an appropriate quotient of the space of $n$-periodic properly bounded difference operators supported in  $[j, j+3]$. Under this identification, closed polygons correspond to those operators whose monodromy is  a scalar multiple of the identity (furthermore, one can arrange that the monodromy of an operator corresponding to a closed polygon is exactly the identity).\par

\begin{remark}
One can adapt the proof of Proposition \ref{freedom} to show that for every $j \in \Z$ there is a one-to-one correspondence between projective equivalence classes of twisted planar $n$-gons and properly bounded $n$-periodic difference operators $\D$ supported in $[j,j+3]$, considered up to the action $\D \mapsto \lambda\circ\D\circ \mu^{-1}$, where $\lambda, \mu \in (\R^*)^\infty$ are sequences of non-zero real numbers which are $n$-quasi-periodic and have the same monodromy. In other words, one has $\lambda, \mu \in (\R^*)^\infty \cap \v{z}$ for some $z \in \R^*$ where the space $\v{z}$ is defined by \eqref{qspace}. We will not use the result in this form in the paper, so we omit the proof.

We also mention that when $n$ is not divisible by $3$, the above described action of quasi-periodic sequences on difference operators admits a section given by $n$-periodic operators of the form
$
T^j + b T^{j+1} + aT^{j+2} - T^{j+3}.
$ This is essentially the content of  \cite[Proposition 4.1]{ovsienko2010pentagram}.
As a result, the entries of periodic sequences $a$, $b$ constitute a global coordinate system on the space of projective equivalence classes of twisted planar $n$-gons. Such a coordinate system is no longer available when $n$ is divisible by~$3$, however our description of the twisted $n$-gons space as a quotient remains valid.
\end{remark}

Dual difference operators correspond to projectively dual polygons. Recall that the dual of a polygon is the polygon in the dual projective plane whose vertices are the sides of the initial one. Note that while there is, in general, no canonical way to label the vertices of the dual polygon, polygons with odd number of vertices admit one particular labeling which is more symmetric than the others. This labeling is depicted in Figure \ref{Fig:dual}. For closed polygons, such labeling makes projective duality an involution. For twisted polygons, it is only an involution up to the action of the monodromy, but still an actual involution on projective equivalence classes.
\begin{definition}\label{def:dual}
Let $P$ be a closed or twisted $n$-gon with odd $n$. Then the $k$'th vertex of its \textit{dual polygon} $P^*$ is the side of $P$ which joins the vertices with indices ${k+(n-1)\,/\,2}$, ${k+(n+1)\,/\,2}$. A polygon $P$ is called \textit{self-dual} if it is projectively equivalent to its dual polygon $P^*$.
\end{definition}
\begin{remark}
Closed self-dual polygons are studied in \cite{fuchs2009self}. In that paper, polygons which are self-dual in the sense of Definition \ref{def:dual} are called $n$-self-dual (where $n$ is the number of vertices).
\end{remark}
With our definition of duality, we have the following:
\begin{proposition}\label{prop:duality}
Let $n$ be odd. Consider a $n$-periodic properly bounded difference operator supported in $[(n-3)/2, (n+3)/2]$
and its dual operator
$
\D^*.
$
Then the polygons corresponding to $\D$ and $\D^*$ are dual to each other in the sense of Definition \ref{def:dual}.
\end{proposition}
\begin{proof}
This follows from more general Proposition \ref{dualdual} in the appendix.
\end{proof}

\subsection{Corner invariants and rescaling}\label{sec:ci}

Another description of the space of polygons modulo projective transformations is by means of so-called \textit{corner invariants}~\cite{schwartz2008discrete, ovsienko2010pentagram}. To every vertex $v_k$ of a polygon, one associates two cross-ratios $x_k, y_k$, as shown in Figure \ref{CI}. The definition of the cross-ratio that we use is
$$
[t_1, t_2,t_3, t_4] := \frac{(t_1 - t_2)(t_3 - t_4)}{(t_1 - t_3)(t_2 - t_4)}.
$$
\begin{figure}[t]
\centering
\begin{tikzpicture}[scale = 1, rotate = -90]
\draw [line width=0.2mm]  (0.7,0.4) -- (1,1) -- (2,1.5) -- (3,1) -- (3.3,0.4);
\fill (0.7,0.4) circle [radius=2pt];
\coordinate [label={$v_{k+2}$}]() at (0.7, 0.4);
\fill (1,1) circle [radius=2pt];
\coordinate [label={45:$v_{k+1}$}]() at (1,1);
\fill (2,1.5) circle [radius=2pt];
\coordinate [label={left:$v_{k}$}]() at (2,1.5);
\fill (3,1) circle [radius=2pt];
\coordinate [label={-45:$v_{k-1}$}]() at (3,1);
\fill (3.3,0.4) circle [radius=2pt];
\coordinate [label={below:$v_{k-2}$}]() at (3.3,0.4);
\fill (2,3) circle [radius=2pt];
\coordinate [label={right:$\bar v_{k}$}]() at (2,3);
\fill (1.4,1.8) circle [radius=2pt];
\coordinate [label={45:$\hat v_{k}$}]() at (1.4,1.8);
\fill (2.6,1.8) circle [radius=2pt];
\coordinate [label={-45:$\tilde v_{k}$}]() at (2.6,1.8);
\draw [dashed] (1,1) -- (2, 3);
\draw [dashed] (1.4, 1.8) -- (2,1.5) ;
\draw [dashed] (2, 3) -- (3,1);
\draw [dashed] (2,1.5) -- (2.6,1.8);
\node at (1.75,6) () {$x_k = [v_{k-2}, v_{k-1}, \tilde v_k, \bar v_k]$};
\node at (2.25,6) () {$y_k = [\bar v_k, \hat v_k, v_{k+1}, v_{k+2}]$};
\end{tikzpicture}
\caption{The definition of corner invariants.}\label{CI}
\end{figure}
\begin{remark}Note that the definition of $x_k, y_k$ requires somewhat more than the $3$-in-a-row condition. However, we do not need to care about this, since in this paper we will only be dealing with weakly convex polygons (see Definition \ref{def:lcp} below) for which the numbers $x_k, y_k$ are well-defined by definition.\end{remark}

Clearly, the sequences $x_k$, $y_k$ only depend on the projective equivalence class of the polygon. Furthermore, in the twisted case these sequences are $n$-periodic, and  $\{ x_1, y_1, \dots, x_n, y_n\}$ is a coordinate chart on an open dense subset of twisted $n$-gons modulo projective transformations. Therefore, since the pentagram map preserves the space of twisted polygons and commutes with projective transformations, it can be written in terms of the $(x,y)$ coordinates. The following formulas are obtained in~\cite{ovsienko2010pentagram}: \begin{proposition} One has
\begin{align}\label{pentFormulas}
x_k' = x_k \frac{1 - x_{k-1}y_{k-1}}{1 - x_{k+1}y_{k+1}},\quad
y_k' = y_{k+1} \frac{1 - x_{k+2}y_{k+2}}{1 - x_{k}y_{k}},
\end{align}
where $x_k,y_k$ are the corner invariants of the polygon $P$, and $x_k', y_k'$ are corner invariants of its pentagram image $P'$.
\end{proposition}
 Here we label $P'$ as in \cite{ovsienko2010pentagram}. The labeling used in Figure \ref{Fig:labeling} leads to the same formulas with a certain shift of indices. Although we will never use these explicit formulas, we will use the following corollary: the pentagram map, with any labeling of vertices, commutes with the $1$-parametric group of transformations $R_s$ given by
\begin{equation}\label{rescaling}
R_s \colon x_k \mapsto sx_k, \quad y_k \mapsto s^{-1}y_k.
\end{equation}
These transformations are known as \textit{rescaling}. They were introduced in \cite{ovsienko2010pentagram} to prove that the pentagram map is a completely integrable system.
\par
We now discuss the relation between two representations of the space of polygons: in terms of difference operators, and in terms of corner invariants.
%
\begin{proposition}\label{prop:xyviad}  Assume we are given a polygon $P$ defined by a difference operator \eqref{o3op} supported in $[j,j+3]$. Then the corner invariants of $P$ are given by \begin{equation}\label{xyviad}
x_{k+ j + 2} = \frac{c_{k}a_{k+1}}{b_{k} b_{k+1}}, \quad y_{k+ j + 2} = \frac{d_kb_{k+1}}{c_{k} c_{k+1}}.
\end{equation}

\end{proposition}
\begin{proof}
The proof is a computation following the lines of the proof of Lemma 4.5 in~\cite{ovsienko2010pentagram}.
\end{proof}
Formulas \eqref{xyviad} allow one to describe rescaling operation  \eqref{rescaling} in terms of difference operators:
\begin{corollary}\label{cor:rescalingDO}
In terms of difference operators, rescaling \eqref{rescaling} can be defined as 
$$
a T^j + bT^{j+1} + cT^{j+2} + dT^{j+3} \mapsto aT^j + bT^{j+1} + s(cT^{j+2} + dT^{j+3}).
$$
\end{corollary}
\begin{proof}
Formulas \eqref{xyviad} show that multiplying $c$ and $d$ coefficients by $s$ is equivalent to multiplying $x$ variables by $s$ and $y$ variables by $s^{-1}$, which is exactly rescaling \eqref{rescaling}.
\end{proof}
\begin{remark}
Note that since there are many operators corresponding to a given polygon, there are also many different ways to define the rescaling on operators. For example, the following formula defines the same operation on polygons as the formula provided above:
$$
a T^j + bT^{j+1} + cT^{j+2} + dT^{j+3} \mapsto a T^j + s^{-{1}/{3}}bT^{j+1} + s^{{1}/{3}}cT^{j+2} + dT^{j+3}.
$$
\end{remark}



\section{Weakly convex polygons}\label{sec:lcp}

In this section we define weakly convex polygons and describe their properties needed to prove Theorems~\ref{thm2} and \ref{thm1}. 
\begin{definition}\label{def:lcp}
A polygon is \textit{weakly convex} if its corner invariants are well-defined and  satisfy $$x_k > 0, \quad  y_k > 0, \quad x_ky_k < 1.$$
\end{definition}
\begin{proposition}\label{xiyi01}
Convex polygons are weakly convex.
\end{proposition}
\begin{proof}
For convex polygons, the collinear points $v_{k-2}$, $v_{k-1}$, $\tilde v_k$, $\bar v_k$ in Figure \ref{CI} are distinct and their cyclic order is exactly  as shown. So, $x_k$ is well-defined and  $0 < x_k< 1$. 
Likewise, we have $y_k \in (0,1)$. The result follows. 
\end{proof}

\begin{remark}
More generally, all corner invariants of a polygon satisfy $x_k, y_k \in (0,1)$ if and only if any five consecutive vertices of that polygon form a convex pentagon (where by a convex pentagon in $\RP^2$ we mean a pentagon which is convex in a suitable affine chart). So, all polygons satisfying this ``5-in-a-row'' condition are weakly convex. 
The geometric meaning of the general weak convexity condition is not that clear. 
However, it turns out to be really convenient for the purposes of the present paper.



\end{remark}




%


The following  is an exhaustive list of properties of weakly convex polygons needed for our purposes:

\begin{proposition}\label{alt}
Assume that $P$ is a closed or twisted weakly convex $n$-gon, where $n \geq 3$ is odd. Then:

\begin{enumerate} \item The corresponding third order $n$-periodic difference operator  \eqref{o3op} can be chosen in such a way that for all $k \in \Z$ we have
\begin{equation}\label{altCond}
a_k, c_k > 0, \quad b_k, d_k < 0.
\end{equation}
\item For any difference operator \eqref{o3op} corresponding to $P$ and satisfying \eqref{altCond}, consider the operators 
$\Dl := a T^j + bT^{j+1}, \Dr := cT^{j+2} + dT^{j+3}.$ Then the monodromies $z_l, z_r$ of these operators (which are real numbers because the operators are of first order) satisfy $0 < z_l < z_r$. In particular, we have $\Ker \Dl \cap \Ker \Dr = 0$.
\item Any polygon obtained from $P$ by means of rescaling \eqref{rescaling} with $s > 0$ is weakly convex. 
\item Any polygon obtained from $P$ by means of rescaling \eqref{rescaling} with $s < 0$ has monodromy with distinct eigenvalues. 
\end{enumerate}
\end{proposition}
\begin{proof}
To prove the first statement, consider the corner invariants $x_k, y_k > 0$ of $P$, and let
$$
\D: = T^j - T^{j+1} + x_{k+j+2}T^{j+2} -  x_{k+j+2}y_{k+j+2}x_{k+j+3}T^{j+3}.
$$
Then, by Proposition \ref{prop:xyviad}, the polygon associated with $\D$ is $P$.
Furthermore, the signs of coefficients of $\D$ satisfy \eqref{altCond}, as needed.\par
To prove the second statement, consider an arbitrary operator  \eqref{o3op} representing $P$ and satisfying~\eqref{altCond}, along with the associated operators $\Dr, \Dl$. Then, by formula~\eqref{monodet}, the monodromies of those operators are given by
$$
z_l = -\prod_{k=1}^n \frac{ a_k}{ b_k}, \quad z_r = -\prod_{k=1}^n\frac{ c_k}{ d_k},
$$
so $z_l, z_r > 0$ due to  \eqref{altCond} and $n$ being odd. Further, using formulas~\eqref{xyviad}, we get
$$  \frac{z_l}{z_r} =  \prod_{k=1}^n \frac{ a _k  d_{k}}{b_{k} c_{k}} = \prod_{k=1}^n x_{k} y_k,
$$
where $x_k, y_k$ are the corner invariants of $P$. But since $P$ is weakly convex, we have $x_ky_k < 1$ and thus $z_l < z_r$. 
 This in turn implies  $\Ker \Dl  \cap \Ker \Dr  = 0$, because non-zero elements of the kernel of $\Dl$ are sequences with monodromy $z_l$, while non-zero elements of the kernel of $\Dr$  are sequences with monodromy $z_r \neq z_l$. Thus, the second statement is proved. \par
The third statement is obvious from the definitions of weak convexity and rescaling, so we proceed to the fourth statement. Let $P_{sd}$ be a polygon obtained from $P$ by means of rescaling with $s < 0$. Then the corner invariants $\hat x_k, \hat y_k$ of $P_{sd}$ satisfy $\hat x_k, \hat y_k < 0$, $\hat x_k\hat y_k < 1$. To show that the monodromy of such a polygon has distinct eigenvalues,
we use a result from the appendix to \cite{izosimov2016pentagrams} which says the monodromy of a twisted $n$-gon with corner invariants $\hat x_k, \hat y_k$ is conjugate to the product $L_1\cdot \ldots \cdot L_n$, where
\begin{equation}\label{li}
L_k :=  \left(\begin{array}{ccc}1 & 0 & 1 \\1-\hat x_k\hat y_k & 0 & 1 \\0 & -\hat y_k & 0\end{array}\right).
\end{equation}
Notice that since $\hat y_k < 0$ and $\hat x_k \hat y_k < 1$, the matrices $L_k$  are non-negative. Furthermore, the product of at least three such matrices is positive, so the matrix $M := L_1\cdot \ldots \cdot L_n$ is positive. Therefore, by the Perron-Frobenius theorem, $M$ has a real positive eigenvalue $z_1$ such that its any other eigenvalue $z$ satisfies $|z| < z_1$. Furthermore, since $\hat x_k < 0$ and $n$ is odd, we have $\det M= \hat x_1 \cdot  \ldots  \cdot \hat x_n(\hat y_1 \cdot \ldots \cdot \hat y_n)^2 < 0$, so the product of two other eigenvalues $z_2, z_3$ of $M$ is negative, which means that they are real and distinct. The result follows.\end{proof}

\section{Polygons fixed by the pentagram map and commuting difference operators}\label{sec:fpcdo}\label{sec:cdo}
In this section we show that a closed or twisted polygon $P$ projectively equivalent to its pentagram image $P'$ gives rise to commuting difference operators. This is the first step in the proof of both Theorem \ref{thm1} and Theorem \ref{thm2}. In addition, in the self-dual case (i.e. in the setting of Theorem \ref{thm2}) we show that those commuting operators are negative duals of each other.


Let $P = \{v_k\}$ be a closed or twisted $n$-gon with odd $n$, and let $P' = \{v_k'\}$ be the image of $P$ under the pentagram map, labeled as in Figure \ref{Fig:labeling}. Then, as explained in Section \ref{sec:polygons}, one can encode $P'$ by means of a difference operator $\D$ of the form 
\begin{equation}\label{diffOp2}
\D = a T^{{(n-3)}/{2}}+ bT^{{(n-1)}/{2}} + cT^{{(n+1)}/{2}} + dT^{{(n+3)}/{2}}.
\end{equation}
The coefficients of this operator are related to the polygon $P'$ by means of the equation
$$
a_k V'_{ i+ {(n-3)}/{2}}+ b_kV'_{k+(n-1)/2} + c_kV'_{k + (n+1)/2} + d_kV'_{k+(n+3)/2} = 0,
$$
where $V'_k$'s are lifts of the vertices $v_k'$ of $P'$.
\begin{proposition}\label{vvp}
The vector $$V_k := a_k V'_{k + {(n-3)}/{2}}+ b_kV'_{k+(n-1)/2} = - c_kV'_{k + (n+1)/2} - d_kV'_{k+(n+3)/2}$$ is the lift of the vertex $v_k$ of $P$.
\end{proposition}
\begin{proof}
Indeed, we have 
$$V_k \in \mathrm{span}(V'_{k + (n-3)/2}, V'_{k+(n-1)/2}) \cap \mathrm{span}(V'_{k + (n+1)/2}, V'_{k+(n+3)/2}),$$
which means that the projection of $V_k$ to $\P^2$ is the intersection point of the lines $(v'_{k + (n-3)/2},v'_{k+(n-1)/2})$ and $(v'_{k + (n+1)/2}, v'_{k+(n+3)/2})$. By definition of the pentagram map with our labeling convention, this is exactly the vertex $v_k$ of $P$, as desired.
\end{proof}

Now, as in Proposition \ref{alt}, consider the operators
\begin{equation}\label{dldr}
\Dl  := a T^{{(n-3)}/{2}}+ bT^{{(n-1)}/{2}}, \quad \Dr  := \D - \Dl  =cT^{{(n+1)}/{2}} + dT^{{(n+3)}/{2}}.
\end{equation}
By Proposition \ref{vvp}, these operators take the lifts $V_k'$ of the vertices of $P'$ to the lifts $\pm V_k$ of the vertices of $P$.
\begin{proposition}\label{prop:cdo}
Assume that a closed or twisted $n$-gon (where $n \geq 5$ is arbitrary) $P$ is projectively equivalent to its pentagram image $P'$. Then:
\begin{enumerate}
\item One can choose the $n$-periodic operator $\D$ of the form \eqref{diffOp2} associated with $P$ in such a way that the corresponding operators $\Dl , \Dr $ given by \eqref{dldr} commute:
\begin{equation}\label{dldrcomm}
\Dl  \Dr  = \Dr  \Dl .
\end{equation}
\item Furthermore, if $P$ is weakly convex and $n$ is odd, then $\D$ can be chosen to satisfy the alternating signs condition \eqref{altCond}.
\item If, on top of that, $P$ is self-dual, then $\Dl$, $\Dr$ may be chosen to be negative duals of each other (up to multiplication by $T^{-n}$): $$\Dl^* = -T^{-n}\Dr.$$ Equivalently, the operator $\D = \Dl + \Dr$ can be chosen to be anti-self-dual (again, up to multiplication by $T^{-n}$): $$\D^* = -T^{-n}\D.$$
\end{enumerate}
\end{proposition}
\begin{proof}[Proof of Proposition \ref{prop:cdo}]
We begin with the first statement. Take an arbitrary $n$-periodic operator~$\tilde \D$ of the form \eqref{diffOp2} representing the polygon $P$. Then, since $P'$ is projectively equivalent to $P$, there is a fundamental solution $V'$ of $\tilde \D$ such that the projection of $V'_k$ to $\P^2$ is the $k$'th vertex of $P'$. Consider the projective transformation taking $P'$ to $P$. Any lift $A \in \GL_3(\R)$ of this projective transformation will then take the sequence $V'$ to a sequence of lifts of vertices of $P$. On the other hand, by Proposition \ref{vvp}, lifts of vertices of $P$ are given by the sequence $\tilde \Dl  V'$. So, there is an $n$-periodic sequence $\mu$ of non-zero real numbers such that
$
AV' =  \mu \tilde \Dl V',
$
where $A$ acts on sequences of vectors by term-wise multiplication. Let $\D:=\mu \tilde \D$. Then the operator $\D$ still satisfies $\D V'=0$ and hence represents the same polygons $P$ and $P'$. Furthermore, the corresponding operator $\Dl $ satisfies
$
AV' =   \Dl V'.
$
Applying the operator $\D$ to both sides, we get 
$
 \D   \Dl  V' = 0.
$
Also taking into account that $ \D V' = 0$, this can be rewritten as
\begin{equation}\label{commutator0}
( \D   \Dl  -  \Dl    \D  )V' = 0.
\end{equation}
At the same time, we have
$$
 \D   \Dl  -  \Dl    \D = ( \Dl  +   \Dr )   \Dl  -  \Dl   ( \Dl  +   \Dr ) =  \Dr    \Dl   -   \Dl    \Dr  = [\Dr , \Dl ],
$$
so \eqref{commutator0} gives
\begin{equation}\label{commutator}
[\Dr , \Dl ]\,V' = 0.
\end{equation}
Now it remains to notice that the commutator $[\Dr , \Dl ]$ is of the form
$
\alpha T^{n-1} + \beta T^n + \gamma T^{n+1},
$
so equation~\eqref{commutator} is equivalent to
$
\alpha_k V'_{k + n -1} + \beta_k V'_{k+n} + \gamma_k V'_{k+n+1} = 0,
$
which, in view of the $3$-in-a-row condition for $P'$ (which holds because it holds for $P$),  gives $\alpha_k = \beta_k = \gamma_k = 0$ and thus
$
[\Dr , \Dl ] = 0,
$
as desired. \par  To prove the second statement, one repeats the same argument, with the only modification that the initial operator $\tilde \D$ should be chosen to satisfy \eqref{altCond}, which can be done by the first statement of Proposition~\ref{alt}. Then the coefficients $a,b,c,d$ of the operator $\D = \mu \tilde \D$  satisfy $\sgn(a_k)  = -\sgn(b_k) = \sgn( c_k) = -\sgn(d_k)$. Furthermore, we claim that $a_k$'s are all of the same sign. Indeed, using explicit formulas \eqref{dldr} for $\Dl $ and $\Dr $ and equating the coefficients of $T^{n-1}$ in \eqref{dldrcomm}, we get
\begin{equation}\label{explicitComm}
a_k c_{k+(n-3)/2} = c_k a_{k+(n+1)/2}.
\end{equation}
Furthermore, we have $\sign c_j = \sign a_j$ for any $j$, so taking the signs of both sides of \eqref{explicitComm} we get
$$
\sign a_{k+(n+1)/2} = \sign c_{k+(n-3)/2} = \sign a_{k+(n-3)/2},
$$
which means that the sequence $\sign a_k$ is $2$-periodic. But since the period $n$ of the sequence $a_k$ is an odd number, it follows that $\sign a_k = \const$. Now, multiplying $\D$ by $-1$ is necessary,  we can arrange that $a_k > 0$ for all $k$, so that $\D$ has satisfies \eqref{altCond}, as needed.

\par

To prove the third statement, we consider the operator $\D$ constructed above and show that if $P$ is self-dual, then $\D$ can be replaced with another operator, which has all the properties of $\D$ and is, in addition, anti-self-dual. To that end, observe that if $P$ is self-dual, then the operators $ \D$ and $ \D^*$ represent the same polygon, so
\begin{equation}\label{sdpop}
 \D^* = \alpha T^{-n } \D \beta^{-1}
\end{equation}
for certain sequences $\alpha, \beta$ of non-zero real numbers. Taking the duals, we get
$$
 \D =  \beta^{-1} T^{n } \D^* \alpha =  \alpha \beta^{-1} \D \alpha\beta^{-1},
$$
which implies $\beta = \pm \alpha$. Further, since $ \D$ satisfies \eqref{altCond}, the corresponding coefficients of $ \D^*$ and $T^{-n } \D$ are of opposite sign, so \eqref{sdpop} implies $\sgn(\alpha_k) = -\sgn(\beta_k) = \const$, and we must have $\beta = -\alpha$. Therefore,
$$
 \D^* = -\alpha T^{-n } \D \alpha^{-1},
$$
where $\sgn(\alpha_k) = \const$ and without loss of generality we can assume $\alpha_k > 0$. Furthermore, since both operators $ \D^*$, $T^{-n } \D$ are $n$-periodic, the sequence $\alpha$ is quasi-periodic, i.e. $\alpha_{k+n} = z\alpha_k$ for some $z \in \R^*$ (actually, $z \in \R_+$). Now, let $\gamma_k := \sqrt{\alpha_k}$. Then the sequence $\gamma$ is also quasi-periodic, so the operator $ \D': = \gamma  \D \gamma^{-1}$ is $n$-periodic. Moreover, it has all the properties of $\D$ and is anti-self-dual. Thus, the proposition is proved.
\end{proof}

\begin{remark}\label{rem:crit}
It is easy to see from the proof that the converse of the first statement is also true: if $\Dr $ and $\Dl $ commute, than $P$ is projectively equivalent to $P'$. For instance, consider the polygon $P$ from Example~\ref{ex:complex}. The vertices of that polygon can be lifted to vectors $V_k := (\lambda^{2k}, \lambda^{3k}, 1)$, where $\lambda := \exp({{2\pi\mathrm{i}}/{7}})$. The sequence $V_k$ is annihilated by a difference operator with constant coefficients, namely by $\D := a T^{{(n-3)}/{2}}+ bT^{{(n-1)}/{2}} + cT^{{(n+1)}/{2}} + dT^{{(n+3)}/{2}}$, where $a,b,c,d \in \C$ are such that the roots of the corresponding characteristic equation $a + bx + cx^2 + d x^3 = 0$ are $\lambda^2$, $\lambda^3$, and $1$. Therefore, the associated operators $\Dr $ and $\Dl $ also have constant coefficients and hence commute. So, the polygon $P$ is indeed projectively equivalent to its pentagram image $P'$.
\end{remark}

\par

\section{The spectral curve}\label{sec:sc}

The results of this section are central to the proof of Theorem \ref{thm2} (and will also be used to derive Theorem \ref{thm1} from Theorem \ref{thm2}). Namely, in Section \ref{ss:genus} we consider the joint spectrum of the commuting difference operators $\Dl$, $\Dr$ constructed above (see Section \ref{sec:cdo}), the so-called \textit{spectral curve}, and show that the genus $g$ of that curve is at most $1$. We note that this estimate on the genus is not predicted by the general theory of commuting difference operators. It seems that the best bound one can get from the general theory is $g \leq 2$. Proving the $g \leq 1$ estimate requires somewhat more careful analysis of the field of meromorphic functions. Also note that even the $g \leq 1$ result is still insufficient to prove Theorem~\ref{thm2}. Another important ingredient of the proof is the so-called \textit{eigenvector function}, which encodes  the joint eigenvectors of the commuting operators  $\Dl$, $\Dr$. We study that function in Section~\ref{ss:evf} and in particular prove that it has very few poles. 

\par

\subsection{The spectral curve and a bound on its genus}\label{ss:genus}
In this section, we construct the spectral curve associated with commuting difference operators $\Dl$, $\Dr$ given by Proposition \ref{prop:sd} and discuss its properties, in particular prove that its genus is at most $1$. 
\begin{remark}

Note that instead of defining the spectral curve using commuting difference operators, we could have done this using the Lax representation, as in \cite{soloviev2013integrability}. However, at the end of the day these two definitions turn out to be equivalent to each other (see Remark \ref{rem:scdef} below). Furthermore, even if we defined the spectral curve using the Lax representation, we would still need commuting difference operators to establish the properties of the curve that we need. So, all in all, these two approaches are equivalent, and our choice is just a matter of convenience.
\end{remark}
\begin{remark}
A different notion of a spectral curve corresponding to a difference operator (and hence a polygon) is defined in  \cite{krichever2015commuting}. In our terminology, it is the spectral curve corresponding to  commuting difference operators $ \hat \D$, $T^n$ where $\hat \D$ is a difference operator supported in $[0,3]$ associated with a given polygon. Since the operator $ \hat \D$ does not, generally speaking, commute with  $\Dl$ and $\Dr$, the corresponding spectral curve seems to have no relation to ours.

\end{remark}

\par

Assume that $P$ is a weakly convex twisted $n$-gon, self-dual and projectively equivalent to its pentagram image. Then, by Proposition {\ref{prop:cdo}}, we have an $n$-periodic operator $\D_l = a T^{{(n-3)}/{2}}+ bT^{{(n-1)}/{2}}$ which commutes with its dual. For notational convenience, we define $$\D_+:=\D_l, \quad \D_- := \D_l^* = -T^{-n}\D_r.$$ Periodicity of these operators means that they also commute with $T^n$. Therefore, we have a whole algebra $\mathcal A$ of commuting operators, generated by $\D_+ $, $\D_-$, and $T^n$ (to preserve the left-right symmetry, it is natural to include $T^{-n}$ in $\mathcal A$ too, so that $\mathcal A= \C[ \D_+, \D_-, T^{\pm n}]$). To such an algebra $\mathcal A$, one can always associate an algebraic curve. This curve may be constructed using any two generic elements of $\mathcal A$. As such elements, we pick the operators $T^n$ and the product $\D_+\D_- = \D_-\D_+$. This choice is motivated by a particularly simple form of the operator $\D_+\D_-$. Namely, that operator is self-dual and supported in $[-1,1]$:
\begin{equation}\label{prodOP}
\D_+\D_- = T^{-1}\alpha + \beta + \alpha T,
\end{equation}
where $\alpha, \beta$ are $n$-periodic sequences, and $\alpha_k \neq 0$ for any $k$. 
\begin{definition}
The \textit{affine spectral curve} $\Gamma_a$ is the joint spectrum of $T^n$ and  $\D_+\D_-$:
$$
\Gamma_a := \{ (z,w) \in \C^* \times \C \mid \exists\, \xi \in \R^\infty: \xi \neq 0, \,T^n \xi = z\xi, \D_+\D_-\xi = w\xi \}.
$$
In other words, a point $(z,w) \in \C^* \times \C$ is in $\Gamma_a$ if and only if $w$ is an eigenvalue of the restriction of $\D_+\D_-$ to the space $\v{z}$ defined by \eqref{qspace}. 
\end{definition}
 To obtain an explicit equation of the affine spectral curve $\Gamma_a$, take a basis $e^{1}, \dots, e^{n}$ in $\v{z}$ determined by the condition $e^{j}_k = \delta_k^j$ for $k = 1,\dots, n$. In this basis, the matrix of the operator $\D_+\D_-$ is almost tridiagonal, with two additional elements in the upper-right and bottom-left corners:
\begin{equation}\label{todaMatrix}
 \left(\begin{array}{cccccc}\beta_1 & \alpha_1 &   & &  \alpha_n z^{-1}  \\ \alpha_1 & \beta_2 & \alpha_2 & \\ & \ddots & \ddots & \ddots \\ &  & \alpha_{n-2} & \beta_{n-1} & \alpha_{n-1}
 \\ \alpha_n z &  & & \alpha_{n-1} & \beta_n\end{array}\right).
\end{equation}

The affine spectral curve $\Gamma_a$ is the zero locus of the characteristic polynomial of~\eqref{todaMatrix}, which, up to the factor $ \alpha_1\ldots\alpha_n$, reads
\begin{equation}\label{charp}
p(z,w) = z^{-1} + q(w) + z
\end{equation}
for a certain polynomial $q(w)$ of degree $n$. In particular, the spectral curve is algebraic, as predicted by the general theory of commuting difference operators.
\begin{proposition}\label{prop:irred}
The affine spectral curve $\Gamma_a$ is irreducible.
\end{proposition}
\begin{proof}
This curve is the zero locus of the polynomial \eqref{charp}, which is irreducible whenever $q(w)$ is non-constant.
\end{proof}

We now define the \textit{spectral curve} $\Gamma$ as the Riemann surface corresponding to the affine curve $\Gamma_a$. In other words, $\Gamma$ is the unique Riemann surface biholomorphic to $\Gamma_a$ away from a finite number of points. The existence of such a Riemann surface is guaranteed by Riemann's theorem. It can be obtained from $\Gamma_a$ by means of normalization (which we actually explicitly construct in Remark \ref{resolution} below), followed by compactification. Since $\Gamma_a$ is irreducible (Proposition~\ref{prop:irred}), it follows that $\Gamma$ is connected. Furthermore, the Riemann surface $\Gamma$ comes equipped with: \begin{itemize} \item two meromorphic functions $z$ and $w$, obtained from coordinate functions on $\Gamma_a$, and satisfying the equation $p(z,w) = 0$, with $p$ given by \eqref{charp}; \item a holomorphic involution $\sigma \colon \Gamma \to \Gamma$, coming from the involution $(z,w) \mapsto (z^{-1},w)$ on $\Gamma_a$, and satisfying  $\sigma^*w = w$ and $\sigma^*z = z^{-1}$. \end{itemize}

\begin{proposition}\label{prop:degrees}
The degrees of the functions $w$ and $z$ on $\Gamma$ are equal to $2$ and $n$ respectively.
\end{proposition}
\begin{proof}
Since the polynomial $q(w)$ in \eqref{charp} has degree $n$, the equation $p(z,w) = 0$ has $n$ solutions in terms of $w$ for generic $z$. So, the degree of $z$ on $\Gamma$ is $n$. Likewise, the number of solutions of $p(z,w) = 0$ in terms of $z$ is $2$ for generic $w$, so the degree of $w$ is $2$.
\end{proof}
Since $w$ is a function of degree $2$, and $\sigma$ is a non-trivial involution preserving $w$, it follows that $\sigma$ interchanges the two points in any level set of $w$. In particular, the fixed points set of $\sigma$ coincides with the set of branch points of $w$, i.e. points where $dw = 0$. At the end of this section we will show that the number of such branch points is at most $4$, which implies, by the Riemann-Hurwitz formula, that the genus of $\Gamma$ is at most $1$. But first we need to discuss in detail the analytic properties of the functions $z$ and $w$, as well as of some other functions on $\Gamma$ which we introduce below.
\begin{proposition}
The Riemann surface $\Gamma$ is obtained from the normalization of $\Gamma_a$ by adding two points $Z_\pm$, interchanged by the involution $\sigma$. The point $Z_+$ is a zero of order $n$ for the function $z$, while $Z_-$ is its pole of order $n$. Both points are simple poles of the function $w$.
\end{proposition}
\begin{proof}

Let $\Gamma_n \subset \Gamma$ be the normalization of $\Gamma_a$. This set can be described as the preimage of $\Gamma_a$ under the map $(z,w) \colon \Gamma \to \P^1 \times \P^1$. Also note that the image of $\Gamma$ under the latter map is precisely the closure of $\Gamma_a$ in $ \P^1 \times \P^1$, which consists of $\Gamma_a$ and the points $(0, \infty)$ and $(\infty, \infty)$. So, the image of $\Gamma \setminus \Gamma_n$ under the map $(z,w)$ is two points $(0, \infty)$ and $(\infty, \infty)$. This means, first, that any point in $\Gamma \setminus \Gamma_n$ is a pole of $w$, and second, that there are at least two such points. But since $w$ has degree $2$ (Proposition~\ref{prop:degrees}), it follows that $\Gamma \setminus \Gamma_n$ consists of exactly two points, and that these points are simple poles of $w$. Furthermore, at one of these points, which we denote by $Z_+$, we have $z = 0$, while at the other one, which we call $Z_-$, we have $z = \infty$. Finally notice that since all points of $\Gamma$ except $Z_\pm$ belong to $\Gamma_n$, it follows that $z$ does not have zeros or poles except for $Z_\pm$. So, $Z_+$ is a zero of $z$ of order $n$, while $Z_-$ is a pole of order $n$, as desired. \end{proof}
Denote also by $S_\pm$ the two zeros of the function $w$ on $\Gamma$. A priori, these two points may coincide, but later on we will show that they are distinct (see the proof of Proposition \ref{cor:table}). Table \ref{table} summarizes information about the orders of the functions $z$ and $w$ at the points $Z_\pm$, $S_\pm$ (recall that the \textit{order} of a meromorphic function $f$ at a point $X$ is equal to $m$ if $f$ has a zero of order $m$ at $X$, $-m$ if $f$ has a pole of order $m$ at $X$, and $0$ otherwise). Also note that the order of $z$ and $w$ at any other point of $\Gamma$ is equal to $0$. The table also contains information about functions $s$ and $\mu_\pm$, which we introduce below.
\begin{table}[t]
\centering
\begin{tabular}{|c|c|c|c|c|c|}\hline Function & Degree & Order at $Z_+$ & Order at $Z_-$  & Order at $S_+$ & Order at $S_-$   \\\hline $z$ & $n$ & $n$ & $-n$ & $0$ & $0$ \\\hline $w$ & $2$ & $-1$ & $-1$ & $1$ & $1$ \\\hline $s$ & $3$ & $-2$ & $2$ & $1$ & $-1$ \\\hline  $\mu_+$ & ${(n-1)/}{2}$ & ${(n-3)/}{2}$ & -${(n-1)/}{2}$ & 1 & 0\\\hline $\mu_-$ & ${(n-1)/}{2}$ & -${(n-1)/}{2}$ & ${(n-3)/}{2}$ & 0 & 1 \\\hline \end{tabular}
\caption{The orders of the functions $z,w,s, \mu_\pm$ at the points $Z_\pm, S_\pm \in \Gamma$. The order of these functions at any other point of $\Gamma$ is zero.}\label{table}
\end{table}
\begin{proposition}\label{prop:rankOne}
The pair $\D_+\D_-$, $T^n$ of commuting difference operators is of \textit{rank $1$}, which means that the generic common eigenspace of these operators is $1$-dimensional.
\end{proposition}
\begin{proof}
As follows from the explicit form  \eqref{charp} for the characteristic polynomial of the matrix \eqref{todaMatrix}, for generic $z$ that matrix has distinct eigenvalues and hence one-dimensional eigenspaces. 
\end{proof}
For a generic point $(z,w) \in \Gamma_a$, let $\xi \in \R^\infty$ be the corresponding common eigenvector of $\D_+\D_-$ and $T^n$, normalized by the condition $\xi_0 = 1$.
\begin{proposition}\label{prop:ev}
The components $\xi_k$ of $\xi$ extend to meromorphic functions on $\Gamma$. The corresponding vector-function $\xi$ on $\Gamma$ satisfies the equations
\begin{equation}\label{xidefeqns}
T^n\xi = z \xi, \quad \D_+\D_- \xi = w \xi.
\end{equation}
\end{proposition}
\begin{proof}
For $(z,w) \in \Gamma_a$, let $\eta = (\eta_1, \dots, \eta_n)$ be the first row of the comatrix of $L - w\Id$, where $L$ is given by~\eqref{todaMatrix}. Extend $\eta$ to a bi-infinite sequence by the rule $\eta_{k+n} = z \eta_k$. Then $\eta$ is a common eigenvector of $T^n$ and $\D_+\D_-$:
$
T^n\eta = z \eta$, $\D_+\D_- \eta = w \eta.
$
Furthermore, the components of $\eta$ are, by construction, rational functions of $z$ and $w$. So, to obtain the desired function $\xi$, it remains to normalize $\eta$:
$
\xi_k = {\eta_k}/{\eta_0}.
$
Note that $\eta_0 = z^{-1}\eta_n$ does not vanish identically on $\Gamma_a$, because $\eta_n$ is a polynomial in $z,w$ which is linear in $z$ and hence cannot be divisible by the defining polynomial of $\Gamma_a$. So, $\xi$ is a well-defined rational vector-function of $z,w$, and hence a meromorphic function on $\Gamma$.
\end{proof}
We call the vector-function $\xi$ the \textit{eigenvector function}. Its analytic properties are studied in detail in the next Section \ref{ss:evf}. 
\begin{remark}\label{rem:holo}
Note that at every point $X \in \Gamma \setminus \{Z_\pm\}$, the vector-function $\xi$ is meromorphic in the following strong sense: there exists a local holomorphic function $f$ such that $f\xi$ is holomorphic at $X$. Moreover, the function $f$ can chosen in such a way that  $(f\xi)(X)$ does not vanish. Therefore, the value of the function $\xi$ at any point $X \in \Gamma \setminus \{Z_\pm\}$ determines a direction in the infinite-dimensional projective space $\mathbb P^\infty$, regardless of whether the components of $\xi$ are finite or infinite (note also that this direction does not change if we replace our particular normalization $\xi_0 = 1$ by any other normalization). This, is however, not true at the points $Z_\pm$. At those points, the components $\xi_k$ of $\xi$ are still meromorphic, but the order of the pole of $\xi_k$ is an unbounded function of $k$ (see Proposition \ref{behinf} below), so there exists no $f$ such that $f\xi$ is holomorphic. In particular, the value of $\xi$ at $Z_\pm$ does not determine any direction.

\end{remark}
We now show that every operator $\mathcal L$ from the commutative algebra $\mathcal A=\langle \D_+ , \D_-,T^{\pm n} \rangle$ gives rise to a meromorphic function $f_{\mathcal L}$ on $\Gamma$, which is holomorphic everywhere except possibly the points $Z_\pm$ and satisfies
$
\mathcal L \xi = f_{\mathcal L} \xi.
$
In particular, the assignment $\mathcal L \mapsto f_{\mathcal L}$ is a homomorphism from $\mathcal A$ to the algebra of meromorphic functions on $\Gamma$ which are holomorphic in $\Gamma \setminus \{Z_\pm\}$. We already have $f_{T^{\pm n}}= z^{\pm 1}$ and $f_{\D_+\D_-} = w$, so it remains to construct the functions $f_{\D_+}$ and $f_{\D_-}$ (of course, one of them determines the other, since their product must be equal to $w$). We denote these functions by $\mu_+$ and $ \mu_-$:
\begin{proposition}\label{musholo}
There exist meromorphic functions $\mu_+, \mu_-$ on $\Gamma$ which are holomorphic in $\Gamma \setminus \{Z_\pm\}$ and satisfy
$
\D_\pm \xi = \mu_\pm \xi.
$
Furthermore, we have $\mu_+\mu_- = w$.
\end{proposition}
\begin{proof}
By Proposition~\ref{prop:rankOne}, a generic common eigenspace of $\D_+\D_-$ and $T^n$ is one-dimensional, and is therefore generated by the vector $\xi$, evaluated at the corresponding point of the Riemann surface $\Gamma$. For this reason, since the operator $\D_+$ commutes with $\D_+\D_-$ and $T^n$, at generic points of $\Gamma$ we must have
\begin{equation}\label{mupluschar}
\D_\pm \xi = \mu_\pm \xi
\end{equation}
for certain numbers $\mu_\pm \in \C$ depending on the point of $\Gamma$.
Furthermore, since the left-hand side of~\eqref{mupluschar} is a meromorphic vector-function on $\Gamma$, and so is $\xi$, it follows that the functions $\mu_\pm$ also extend to meromorphic functions on the whole of $\Gamma$. Moreover, given a point $X \in \Gamma \setminus \{Z_\pm\}$, renormalizing $\xi$ if necessary we can assume that $\xi(X)$ is finite and non-zero (see Remark \ref{rem:holo}). But then \eqref{mupluschar} implies that the functions $\mu_\pm$ are holomorphic at $X$. Finally, the equation $\mu_+\mu_- = w$ follows directly from \eqref{mupluschar} and the second of equations \eqref{xidefeqns}.
\end{proof}

\begin{proposition}\label{sigmamumu}
We have $\sigma^* \mu_+ = \mu_-$.
\end{proposition}
\begin{remark}
The existence of the involution $\sigma$ on $\Gamma$ is due to the invariance of the algebra $\mathcal A=\langle \D_+ , \D_-,T^{\pm n} \rangle$ under operator duality: $\mathcal A = \mathcal A^* := \{ \mathcal L^* \mid \mathcal L \in \mathcal A\}$. So, since $\D_+^* = \D_-$, it is only natural that $\sigma^*\mu_+ = \mu_-$.
\end{remark}
\begin{proof}[Proof of Proposition \ref{sigmamumu}]
It suffices to show that $\mu_+(X_+) = \mu_-(X_-)$,  where $X_\pm$ is a generic pair of points interchanged by $\sigma$. Since the points $X_\pm$ are generic, one can assume that the vectors $\xi(X_\pm)$ are finite. Under this assumption, we have
$
\D_\pm \xi(X_\pm) = \mu_\pm(X_\pm) \xi(X_\pm).
$
Furthermore, we have $ \xi(X_\pm) \in \v{z_+^{\pm 1}}$, where $z_+ := z(X_+) = z(X_-)^{-1}$. So, using the pairing~\eqref{pairing} between $\v{z_+}$ and $\v{z_+^{-1}}$, we get
\begin{equation}\label{mumustar}
\mu_+(X_+)\left\langle\xi(X_+), \xi(X_-) \right\rangle =  \left\langle \D_+ \xi(X_+), \xi(X_-) \right\rangle =  \left\langle  \xi(X_+), \D_- \xi(X_-) \right\rangle = \mu_-(X_-)\left\langle\xi(X_+), \xi(X_-)\right \rangle.
\end{equation}
So, to complete the proof, it suffices to show that $\left\langle\xi(X_+), \xi(X_-) \right\rangle \neq 0$.
 To that end, observe that $\xi(X_+)$ belongs to the kernel of the operator $(\D_+\D_- - w_0)\vert_{\v{z_+}}$, where $w_0 := w(X_\pm)$, and, in the generic case, spans that kernel. So, the orthogonal complement to $\xi(X_+)$ with respect to the pairing~\eqref{pairing} is the image of the dual operator 
 $$((\D_+\D_- - w_0)\vert_{\v{z_+\vphantom{z_+^-1}}})^* = \left(\D_+\D_- - w_0\right)\vert_{\v{z_+^{-1}}}.$$
 But for generic $z_+$ the operator $\D_+\D_-$ has simple spectrum on $\v{z_+^{-1}}$ and is, therefore, diagonalizable, which in particular implies $$\Im (\D_+\D_- - w_0)\vert_{\v{z_+^{-1}}} \cap \Ker (\D_+\D_- - w_0)\vert_{\v{z_+^{-1}}} = 0.$$ Therefore, we have $\left\langle\xi(X_+), \xi(X_-) \right\rangle \neq 0$, as desired. \end{proof}

Now, define a meromorphic function $s$ on $\Gamma$ by the formula
\begin{equation}\label{SDEF}
s:=\frac{\mu_+}{z\mu_-}.
\end{equation}
This function does not correspond to any difference operator $\mathcal L \in \mathcal A = \C[ \D_+, \D_-, T^{\pm n}]$ but can be thought of as corresponding to a pseudo-difference operator $T^{-n}\D_+\D_-^{-1}$. Accordingly, the function $s$ satisfies the equation
\begin{equation}\label{SEQN}
(\D_+ - sT^n \D_-)\xi = 0.
\end{equation}
Recall that the operator on the left-hand side encodes the family of polygons obtained from $P$ by means of rescaling \eqref{rescaling}. 
\begin{proposition}\label{degIneq}
The function $s$ has degree $3$. There exist three distinct points on $\Gamma$ at which $s = -1$. The function $z$ takes three distinct values at those points.

\end{proposition}
\begin{proof}
We first show that $\deg s \leq 3$. Let $X_1, \dots, X_m \in \Gamma$ belong to the level set $s = s_0$. Then, for generic $s_0 \in \C$, these points correspond to distinct points on the affine spectral curve $\Gamma_a$. This, in particular, means that the vectors $\xi(X_1), \dots, \xi(X_m) \in \Ker(\D_+ - s_0T^n \D_-)$ are linearly independent (as joint eigenvectors of $T^n$ and $\D_+\D_-$ corresponding to distinct eigenvalues). But $$\dim \Ker(\D_+ - s_0T^n \D_-) = \ord(\D_+ - s_0T^n \D_-)= 3,$$ so $m \leq 3$, and the degree of $s$ is at most $3$, as desired.\par
We now show that  there exist three distinct points on $\Gamma$ at which $s = -1$, which, in turn, implies that the degree of $s$ is exactly $3$. Since the polygon $P$, given by the operator $\D_+ -  T^n \D_-$, is weakly convex, it follows from the fourth statement of Proposition \ref{alt} that the monodromy of $\D_+ + T^n \D_-$ has simple spectrum. This means that there exist three distinct numbers $z_1, z_2, z_3$ such that the operator $\D_+ +T^n \D_-$ has non-trivial (and hence one-dimensional) kernel on $\v{z_k}$. Let $\xi^{k}$ be the generator of that kernel. Then, since the operator $\D_+  \D_-$ commutes with $\D_+ +T^n \D_-$ and $T^n$, it follows that $\xi^{k} $ is also an eigenvector of $\D_+  \D_-$, corresponding to some eigenvalue $w_k$. Then the three points $(z_k,w_k)$ belong to the affine spectral curve $\Gamma_a$, and thus give rise to at least three points $X_1, X_2, X_3 \in \Gamma \setminus \{Z_\pm\}$ with $z(X_k) = z_k$, $w(X_k) = w_k$. We now claim that $s(X_1) = s(X_2) = s(X_3) = -1$. Indeed, the vector $\xi_k$ spans the $(z_k,w_k)$ joint eigenspace of $T^n$ and  $\D_+  \D_-$. Therefore, at each of the points $X_k$, we have $\xi(X_k) = c_k\xi^{k}$, where $c_k \in \C$ (here we assume here that the vectors $\xi(X_k)$ are finite, which can be always arranged by multiplying $\xi$ by an appropriate meromorphic function, see Remark \ref{rem:holo}). So, by construction of the vectors $\xi^{k}$ we have

$$
(\D_+ +T^n \D_-) \xi(X_k) =  c_k(\D_+ +T^n \D_-) \xi^{i}= 0.
$$
On the other hand,
$$
(\D_+ +T^n \D_-) \xi(X_k) = (\mu_+ + z\mu_-)\vert_{X_k} \xi(X_k),
$$
so \begin{equation}\label{smone}(\mu_+ + z\mu_-)\vert_{X_k} = 0.\end{equation} 
Notice also that $X_k$ cannot be a common zero of $\mu_+$ and $\mu_-$, because that would imply $\xi(X_k) \in \Ker \D_+ \cap \Ker \D_-$, which is not possible by the second statement of Proposition \ref{alt} (the latter applies to $\D_\pm$ since $\D_+ = \Dl$, $\D_- = -T^{-n} \D_r$). Furthermore, $\mu_\pm$ cannot have a pole at $X_k$ by Proposition~\ref{musholo}. But then~\eqref{smone} implies $s(X_k) = -1$, as desired.  
\end{proof}

\begin{remark}\label{rem:scdef}
It follows from Proposition \ref{degIneq} that the function $s$ has the following meaning. Fix some generic $s_0 \in \C$. Then there are three points $X_1, X_2, X_3$ in $\Gamma$ with $s = s_0$. Furthermore, the vectors $\xi(X_k) \in \v{z(X_k)}$ belong to  $ \Ker(\D_+ - s_0 T^n \D_-) $. So, $z(X_1), z(X_2), z(X_3)$ is the spectrum of the monodromy of $\D_+ - s_0T^n \D_-$. In other words, if we consider a meromorphic mapping $\Gamma \to \C^2$ given by the functions $(z,s)$, then its image belongs to the algebraic curve
$$
\Gamma_a' := \{ (z,s) \in \C^* \times \C \mid z \mbox{ is an eigenvalue of the monodromy of }  \D_+ - sT^n \D_- \}.
$$
Using also that $\deg z = n$ and $\deg s = 3$, it is easy to show that the mapping $\Gamma \to \Gamma_a'$ is generically biholomorphic. So, $\Gamma_a'$ is just another affine model of the spectral curve $\Gamma$. This model can be thought of as the joint spectrum of the operators $T^n$ and $T^{-n} \D_+\D_-^{-1}$ (the latter is well-defined on a generic eigenspace of $T^n$). Furthermore, since the operator $ \D_+ - sT^n \D_-$ corresponds to the polygon $R_s( P_{})$, where $R_s$ is the rescaling action \eqref{rescaling}, it follows that $\Gamma_a'$ can be regarded as the graph of the spectrum for the monodromy of $R_s( P_{})$. As explained in \cite{ izosimov2016pentagrams}, this definition of the spectral curve coincides with the one used in \cite{soloviev2013integrability} to prove algebraic integrability of the pentagram map. So, as a Riemann surface, our spectral curve is isomorphic to the one of  \cite{soloviev2013integrability}.

\end{remark}
We are now in a position to prove the main result of the section:
\begin{proposition}\label{genus2}
The genus $g$ of $\Gamma$ satisfies $g \leq 1$.
\end{proposition}
\begin{proof}
The function $w$ on $\Gamma$ is a $2$-fold ramified covering of the Riemann sphere whose branch points coincide with fixed points of the involution $\sigma$. To estimate the number of such fixed points, notice that from Proposition~\ref{sigmamumu} and formula~\eqref{SDEF} it follows that $\sigma^*s = s^{-1}$. Thus, at each fixed point of $\sigma$ we must have $s = \pm 1$. Furthermore, by Proposition~\ref{degIneq}, the function $z$ takes three distinct values at points of $\Gamma$ where $s = -1$, and since the set $s = -1$ is invariant under the involution $\sigma$ which takes $z$ to $z^{-1}$, it follows that those values must be of the form $\pm 1, z_0, {z_0}^{-1}$, where $z_0 \neq \pm 1$. So, $\sigma$ must have exactly one fixed point at the level set $s = -1$. In addition to that, it may have up to three fixed points at the level set $s = 1$, which is up to four fixed points in total. Now, the desired inequality for the genus follows from the Riemann-Hurwitz formula.
\end{proof}
\begin{remark}\label{z1}
In fact, since the values of $z$ at points where $s = -1$ are the eigenvalues of the monodromy of $\D_+ +T^n \D_-$, it follows from formula \eqref{monodet} that they are of the form $-1, z_0,z_0^{-1}$. Another way to see this is to notice that by Proposition \ref{sigmamumu} at fixed points of $\sigma$ we must have $\mu_+ = \mu_-$ and thus $z = s^{-1}$ (here we use that the functions $\mu_\pm$ do not have common zeros and also do not have poles in $\Gamma \setminus \{Z_\pm\}$). This also implies that if $\Gamma$ has genus $1$, then $z = 1$ at points where $s = 1$. In other words, all eigenvalues of the monodromy of the polygon $P$ are equal to $1$. Later on, we will see that this monodromy is in fact the identity. In other words, if the spectral curve is elliptic, then in the setting of Theorem \ref{thm2} the polygon $P$ must be closed (see Remark \ref{psdclosed}).
\end{remark}
\begin{remark}
Note that without weak convexity (used to prove Proposition \ref{degIneq}) we would not be able to say that there is just one fixed point of the involution $\sigma$ at the level $s = -1$. In that case, nothing seems to prevent $\sigma$ from having six fixed points, which  means that $\Gamma$ may be  a genus $2$ curve. Thus, it should be possible to construct a counterexample to Theorems \ref{thm1} and \ref{thm2} in the non-weakly-convex case using genus $2$ curves and their associated genus $2$ theta functions.  \par Another way to obtain the estimate $g \leq 2$ is to use the existence on $\Gamma$ of meromorphic functions of degree $2$ and $3$ (namely, $w$ and $s$). However, this does not guarantee the $g \leq 1$ estimate obtained above.
\end{remark}
We finish this section with two additional results on the spectral curve which will be useful later on.
\begin{proposition}\label{cor:table}
The degrees of the functions $\mu_\pm, s$ and their orders at points $Z_\pm, S_\pm$ are as shown in Table \ref{table}.
\end{proposition}
\begin{proof}
Since for any $\lambda \in \C^*$ the degree of operators $\D_\pm - \lambda $ is $(n-1)/2$, an argument analogous to the one we used to show that $\deg s \leq 3$ (see the proof of Proposition \ref{degIneq}) gives $\deg \mu_\pm \leq (n-1)/2$. Further, let $d_\pm := \ord_{Z_+}\mu_\pm$. Then, since $\ord_{Z_+}w=-1$ (see Table \ref{table}), the equation $\mu_+\mu_- = w$ implies
\begin{equation}\label{sumofks}
d_+ + d_- = -1.
\end{equation}
Furthermore, using that $\ord_{Z_+}z=n$ and equation \eqref{SDEF}, we get
\begin{equation}\label{diffofks}
d_+ - d_- = n + \ord_{Z_+}s,
\end{equation}
so
\begin{equation}\label{kminfla}
d_- = -\frac{1}{2}(n+1) - \frac{1}{2}\ord_{Z_+}s,
\end{equation}
and since $\deg \mu_- \leq (n-1)/2$, we must have $d_- \geq -  (n-1)/2 $, which implies $\ord_{Z_+}s \leq - 2$. On the other hand, we know that $\deg s = 3$, and from \eqref{kminfla} it follows that $\ord_{Z_+}s$ is even. So, we must have $\ord_{Z_+}s =- 2$, which, along with \eqref{kminfla} implies $\ord_{Z_+}\mu_- = d_- = -  (n-1)/2$ and thus $ \deg \mu_- =  (n-1)/2$. Similarly, adding up \eqref{sumofks} and \eqref{diffofks}, we get $ \ord_{Z_+}\mu_+ = d_+ =  (n-3)/2$. Analogously, replacing the point $Z_+$ with $Z_-$, we find the orders of $\mu_\pm$ and $s$ at $Z_-$, as well as the degree of~$\mu_+$ (one can also use that $\sigma^*\mu_+ = \mu_-$ and $\sigma(Z_+) = Z_-$).\par

It now remains to find the orders of the functions $\mu_\pm$ and $s$ at the points $S_\pm$. To that end, we first show that $S_+ \neq S_-$. Assume, for the sake of contradiction, that $S_+ = S_- = S$. Then $S$ is a double zero of the function $w$. Furthermore, we have $\mu_+\mu_- = w$, and both $\mu_+$ and $\mu_-$ are holomorphic and have at worst a simple zero at $S$ (indeed, these functions have degree $ (n-1)/2 $ and zeros of order $(n-3)/2$ at the points $Z_+$ and $Z_-$ respectively). So, both $\mu_+$ and $\mu_-$ must have a simple zero at $S$. But then, from the definition \eqref{SDEF} of the function $s$ it follows that it does not have a zero at $S$. Furthermore, $s$ cannot have zeros at other points of $\Gamma \setminus \{Z_\pm\}$, because the only zero of $\mu_+$ in that domain is the point $S$. But this means that $s$ has just two zeros counting with multiplicities, which is impossible since $\deg s = 3$. Therefore, we must have that $S_+ \neq S_-$.\par
Now, the relation $\mu_+\mu_- = w$ implies that at both points $S_\pm$ one of the functions $\mu_\pm$ has a simple zero, while the second one does not have a zero or a pole. Without loss of generality, assume that $\mu_+(S_+) = 0$ and thus $\ord_{S_+}\mu_+ = 1$. Then $\ord_{S_+}\mu_- = 0$, and by formula~\eqref{SDEF} we get $\ord_{S_+}s = 1$. Furthermore, from $\sigma^*\mu_+ = \mu_-$ and $\sigma(S_+) = S_-$ it follows that $\ord_{S_-}\mu_+ = 0$, $\ord_{S_-}\mu_- = 1$, and $\ord_{S_+}s = -1$. Finally, notice that the functions $\mu_\pm$ and $s$  do not have zeros or poles other than the points $Z_\pm, S_\pm$, because for each of them the total number of zeros and poles at those points (counting with multiplicities) coincides with the degree. Thus, the proposition is proved.
\end{proof}

\begin{proposition}\label{prop:nodal}
The affine spectral curve $\Gamma_a$ is a nodal curve (i.e. all its singularities are double points).
\end{proposition}
\begin{proof}
The affine spectral curve $\Gamma_a$ is the zero locus of the characteristic polynomial $p(z,w) = z + z^{-1} + q(w)$ of the matrix~\eqref{todaMatrix}. Computing the differential of that polynomial, we get that $(z_0,w_0) \in \Gamma_a$ is singular if and only if $z_0 = \pm 1$ and $w_0$ is a multiple root of $p(z_0,w)$. Furthermore, computing the Hessian, we get that a singular point $(z_0,w_0) \in \Gamma_a$ is a double point if and only if $w_0$ is a double root of $p(z_0,w)$. But for $z_0 = \pm 1$ the matrix ~\eqref{todaMatrix} is symmetric (equivalently, the restriction of the operator $\D_+\D_-$ to $\v{\pm 1}$ is self-adjoint), so the multiplicity of the root $w_0$ of its characteristic polynomial is equal to the dimension of the corresponding eigenspace, which is
$$\dim \Ker(\D_+\D_- - w_0)\vert_{\v{\pm 1}} \leq \ord(\D_+\D_- - w_0) = 2. $$
So indeed all singular points of $\Gamma_a$ are double points. 
\end{proof}

\begin{remark}\label{nodal}

It is easy to see that the genus of the normalization $\Gamma$ of $\Gamma_a$ is equal to $n - d - 1$ where $d$ is the number of double points of $\Gamma_a$. Furthermore, as can be seen from the proof of Proposition~\ref{prop:nodal}, double points of $\Gamma_a$ correspond to double roots of the polynomials $q(w) \pm 2$. The polynomial $q$ is of degree $n$, so each of those polynomials may have at most $(n-1) / 2$ double roots. Therefore, $\Gamma$ is rational when each of the polynomials $q(w) \pm 2$ has precisely  $(n-1) / 2$ double roots (and, in addition, one simple root). Likewise, $\Gamma$ is elliptic when one of the polynomial $q(w)  \pm 2$ has $(n-1)/2$ double roots, while the second one has $(n-3)/2$ double roots. Using Remark \ref{z1} one can show that it is the polynomial $q(w) +2$ that has $(n-3)/2$ double roots.

\end{remark}
 
 \par
\subsection{The eigenvector function}\label{ss:evf}
In this section, we study in detail the analytic properties of the meromorphic vector-function $\xi$ constructed in Proposition \ref{prop:ev}. This will allow us to obtain analytic formulas for coordinates of vertices of the polygon $P$ (see Section \ref{sec:rat} for the rational case and Section \ref{ss:sd} for the elliptic case). We keep all the notation of Section \ref{ss:genus}.
\begin{proposition}\label{behinf}
We have  $\mathrm{ord}_{Z_\pm}\xi_k = \pm k$. 
\end{proposition}
\begin{proof}

Let $d_k := \mathrm{ord}_{Z_+}\xi_k - k$. We need to show that $d_k = 0$ for every $k \in \Z$. Note that $d_0 = 0$ since $\xi_0 = 1$. So it suffices to show that $d_k$ is a constant sequence. Also note that since $\xi_{k+n} = z\xi_k$ and $z$ has a zero of order $n$ at $Z_+$ (see Table \ref{table}), the sequence $d_k$ is $n$-periodic. So, if it is not constant, then there must exist $k \in \Z$ such that $d_{k-1} > d_k \leq d_{k+1}$.  
But since $\xi$ is the eigenvector of the operator \eqref{prodOP} with eigenvalue $w$, we have
\begin{equation}\label{xirel}
 \alpha_{k-1}\xi_{k-1}  + \alpha_{k}\xi_{k+1} = ( w - \beta_k)\xi_k. 
\end{equation}
Since $\alpha$ is a non-vanishing sequence, the order  of the left-hand side at $Z_+$ can be bounded as
\begin{equation}\begin{gathered}
\mathrm{ord}_{Z_+}( \alpha_{k-1}\xi_{k-1}  + \alpha_{k}\xi_{k+1}) \geq  \min(\mathrm{ord}_{Z_+}\xi_{k-1},\mathrm{ord}_{Z_+}\xi_{k+1})  \\ = \min(d_{k-1} + k -1, d_{k+1} + k + 1) \geq d_k + k.
\end{gathered}\end{equation}
On the other hand, since $\mathrm{ord}_{Z_+}w = -1$, the order of the right-hand side of \eqref{xirel} is $d_k + k - 1$. So, $d_k$ must be a constant sequence, as desired.
\end{proof}
We now proceed to describe the behavior of $\xi$ away from the points $Z_\pm$. We begin with the following preliminary lemma.
\begin{lemma}\label{indep}
Assume that $X_\pm \in \Gamma \setminus \{Z_\pm\}$ are distinct points such that $w(X_+) = w(X_-)$ (equivalently, $\sigma(X_+) = X_-$). Then the directions (i.e. points in $\P^\infty$, see Remark \ref{rem:holo}) determined by the values of $\xi$ at $X_\pm$ are distinct from each other.
\end{lemma}
\begin{proof}
Without loss of generality, assume that the vectors $\xi(X_\pm)$ are finite and non-zero (if not, we multiply $\xi$ by an appropriate meromorphic function, see Remark \ref{rem:holo}). One then needs to show that these vectors are linearly independent. To that end, recall that $T^n\xi(X_\pm) = z(X_\pm)\xi(X_\pm)$. So, if $z(X_+) \neq z(X_-)$, then the vectors $\xi(X_\pm)$ are independent as eigenvectors of $T^n$ corresponding to distinct eigenvalues. Therefore, it suffices to consider the case $z(X_+) = z(X_-)$. In that case, we have
$$
z(X_+) = z(X_-) = z(\sigma(X_-))^{-1} = z(X_+)^{-1},
$$
so $z(X_\pm) = \pm 1$. Suppose for the sake of contradiction that the corresponding vectors $\xi(X_\pm)$ are linearly dependent. Then, without loss of generality, we can assume that $\xi(X_+) = \xi(X_-)$ (this can be always arranged by multiplying $\xi$ by an appropriate meromorphic function). Denote $\xi_0:=\xi(X_\pm)$, $z_0 := z(X_\pm) = \pm 1$, $w_0 := w(X_\pm)$. Notice that since $w(X_+) = w(X_-)$, the differential of $w$ does not vanish at $X_\pm$, so $w$ can be taken as a local parameter near those points. Then, differentiating the relation $(T^n-  z)\xi = 0$ with respect to $w$ at $X_\pm$, we get
$$
(T^n - z_0)\xi'(X_+) = z'(X_+)\xi_0, \quad (T^n - z_0)\xi'(X_-) = z'(X_-)\xi_0.
$$
Taking a linear combination of these equations, we obtain
$
(T^n - z_0)\hat \xi = 0,
$
where $
\hat \xi := z'(X_+)\xi'(X_-) - z'(X_-)\xi'(X_+).
$
In other words, we have 
$\hat \xi \in \v{z_0}.$ Similarly, using the equation $(\D_+\D_- -  w)\xi = 0$, we get
\begin{equation}\label{Jordan}
(\D_+\D_- -  w_0)\hat \xi =\lambda \xi_0,
\end{equation}
where
$
\lambda :=z'(X_+) -z'(X_-).
$
Note also that $\lambda \neq 0$. Indeed, $\lambda = 0$ would mean that the two branches of the curve $\Gamma_a$ given by the functions $z(w)$ near $X_\pm$ are tangent to each other. This is however, not possible, since $\Gamma_a$ is a nodal curve (Proposition \ref{prop:nodal}). So, since $\lambda \neq 0$ and $\hat \xi \in \v{z_0}$, it follows from~\eqref{Jordan} that the operator $(\D_+\D_-)\vert_{\v{z_0}}$ has a non-trivial Jordan block. This is, however, not possible, since $z_0 = \pm 1$, and thus $\D_+\D_-$ is self-adjoint on $\v{z_0}$. So it must be that the vectors $\xi(X_\pm)$ are linearly independent, as desired.
\end{proof}
\begin{remark}\label{resolution}
In the elliptic case (i.e. when the genus of $\Gamma$ is $1$), one can also prove Lemma~\ref{indep} as follows. The vectors $\xi(X_\pm)$ are common eigenvectors of the operators $\D_+$, $\D_-$, $T^n$, with the corresponding eigenvalues given by the values of the function $\mu_+$, $\mu_-$, $z$ at the points $X_\pm$. So, it follows that $\xi(X_\pm)$ are independent as long as at least one of the functions $\mu_+, \mu_-, z$ separates $X_+$ from $X_-$. Assume that this is not the case, which means that $\mu_\pm(X_+) = \mu_\pm(X_-)$ and $z(X_+) = z(X_-)$. Then we have
$$\mu_-(X_+) = \mu_+(\sigma(X_+)) = \mu_+(X_-) = \mu_+(X_+).
$$
Along with $z(X_\pm) = \pm 1$, this gives $s(X_\pm) = \pm 1$. But in the elliptic case four of the six points where $s = \pm 1$ are fixed by $\sigma$ (which would force $X_+ = X_-$), while the remaining two points are separated by the function $z$ (by Proposition \ref{genus2}). So indeed the functions $\mu_\pm, z$ separate any pair of points on $\Gamma$, which proves Lemma \ref{indep}. 
As a byproduct, we also get the following result: the functions $\mu_\pm, z$ define an embedding $\Gamma \setminus \{Z_\pm\} \hookrightarrow \C^3$. In other words, if we view $\mu_\pm$ and $z$ as rational functions on $\Gamma_a$, then these functions provide a resolution of singularities.
\end{remark}

Lemma \ref{indep} also admits an infinitesimal version, corresponding to the case when $X_+ = X_- = X$ is a branch point of $w$ (equivalently, a fixed point of $\sigma$). In this case, the role of $\xi(X_\pm)$ is played by the vectors $\xi(X)$, $\xi'(X)$, where the derivative is taken with respect to a local parameter near $X$. Note that upon renormalization of $\xi$, its derivative changes as $\xi' \mapsto f\xi' + f'\xi$, so the direction of $\xi'$ is well-defined modulo the direction of $\xi$. In particular, linear independence of $\xi$ and $\xi'$ is well-defined.
\begin{lemma}\label{indep2}
Assume that $X \in \Gamma$ is a branch point of $w$ (equivalently, a fixed point of $\sigma$). Then the directions determined by the values of $\xi$ and $\xi'$ at $X$ are distinct from each other.
\end{lemma}
\begin{proof}
Renormalizing $\xi$ if necessary, we can assume that its value at $X$ is finite and non-zero. Then, differentiating the equation $(T^n - z)\xi = 0$ with respect to a local parameter near $X$, we get \begin{equation}\label{JBlock}(T^n - z(X))\xi'(X) = z'(X) \xi(X).\end{equation}
Also note that since $\Gamma_a$ is a nodal curve, it follows that the mapping $(z,w) \colon \Gamma \setminus \{Z_\pm\} \to \C^2$ is an immersion, so at a branch point of $w$ me must have $z' \neq 0$. But then \eqref{JBlock} implies that the vectors $\xi(X)$ and $\xi'(X)$ are linearly independent, as desired.
\end{proof}
\begin{remark}\label{dersol}
Differentiating $(\D_+\D_- -  w)\xi = 0$ at $X$ and using that $w'(X) = 0$, we get $(\D_+\D_- -  w(X))\xi'(X) = 0$. So, Lemma \ref{indep2} means that $\xi(X)$ and $\xi'(X)$ form a basis of solutions for the equation $(\D_+\D_- -  w(X))\xi = 0$.
\end{remark}
\begin{proposition}\label{gpoles}
The function $\xi_1$ has $g$ poles in $\Gamma \setminus \{Z_\pm\} $, where $g \in \{0,1\}$ is the genus of $\Gamma$.
\end{proposition}
\begin{proof}
Let $u \in \bar \C$, and let $X_\pm$ be the two preimages of $u$ under the function $w \colon \Gamma \to \bar \C$. Recall that the \textit{trace of a meromorphic function} $f$ on $\Gamma$ under $w$ is a meromorphic function on $\bar \C$ is defined by $(\tr_wf)(u) := f(X_+) + f(X_-) $. Let 
$$
\zeta(u) := \left|\!\begin{array}{cc}\xi_0(X_+) & \xi_1(X_+) \\ \xi_0(X_-) & \xi_1(X_-)\end{array}\!\right|^2. 
$$
Note that $\xi_0 \equiv 1$ by definition of the eigenvector function $\xi$, so
$
\zeta(u) =  (\xi_1(X_+) - \xi_1(X_-))^2.
$
This means that $\zeta := 2 \tr_w (\xi_1^2) - (\tr_w \xi_1)^2$, so in particular $\zeta$ is meromorphic (i.e. rational). To understand the behavior of that function, fix a point $u_0 \in \bar \C$. Let $\Sigma := w(\{X \in \Gamma \mid dw(X) = 0\})\subset \bar \C$ be the set of critical values of $w$ (this set contains two or four points depending on the genus of $\Gamma$). Then the following cases are possible: \\ \\
\textbf{Case 1.} $u_0 \notin \Sigma$ is finite, and $\xi_1$ is finite at both preimages $X_\pm$ of $u_0$ under $w$. In this case, $\zeta(u_0)$ is the squared Wronskian of the solutions $\xi(X_\pm)$ of equation $\D_+\D_-\eta = u_0\eta$. By Lemma \ref{indep}, these solutions are independent, so $\zeta(u_0)$ is finite and non-zero. \\ \\
\textbf{Case 2.}  $u_0 \notin \Sigma$ is finite, $\xi_1$ has a pole of order $d$ at one of the preimages $X_\pm$ of $u_0$ (say, $X_+$), and is finite at the other preimage. 
In this case, the function $(u-u_0)^{2d}\zeta(u)$ is finite at $u_0$ and is equal to the squared Wronskian of linearly independent solutions $((w - u_0)^d\xi)({X_+})$, $\xi(X_-)$ of  $\D_+\D_-\eta = u_0\eta$. So, $\zeta$ has a pole of order $2d$ at $u_0$.\\ \\
\textbf{Case 3.}  $u_0 \notin \Sigma$ is finite, and $\xi_1$ has poles at both preimages $X_\pm$ of $u_0$. This is not possible, since after renormalizing $\xi$ we would get
$$
\left(\!\begin{array}{cc}\xi_0(X_+) & \xi_1(X_+) \\ \xi_0(X_-) & \xi_1(X_-)\end{array}\!\right) = \left(\begin{array}{cc}0 & 1 \\ 0& 1\end{array}\right),
$$
which would mean that the Wronskian of $\xi(X_\pm)$ vanishes. \\ \\
\textbf{Case 4.}  $u_0 = \infty$ (in which case we also have $u \notin \Sigma$). In this case $X_\pm = Z_\pm$, so $\zeta$ has a pole of order $2$ at $u_0$ by Proposition \ref{behinf}. \\ \\
All in all, the function $\zeta$ does not vanish in $\bar \C \setminus \Sigma$, while the number of its poles in that domain is twice the number  of poles of $\xi_1$ in $\{X \in \Gamma \mid dw(X) \neq 0\}$ (counting with multiplicities). Now, consider $u_0 \in \Sigma$, and let $X \in \Gamma$ be the unique point such that $w(X) = u_0$. Then there exists a parameter $t$ near $X$ such that the function $w$ can be locally written as  $t \mapsto u_0 + t^2$. So $\zeta(u)$ near $u_0$ can be written as
$$
\zeta(u) =  \left|\!\begin{array}{cc}\xi_0(t) & \xi_1(t) \\ \xi_0(-t) & \xi_1(-t)\end{array}\!\right|^2,
$$
where $t = \sqrt{u - u_0}$. Then at $t = 0$ we have
\begin{equation}\label{asym}
\zeta(u) \sim  t^2\left|\!\begin{array}{cc}\xi'_0(0) & \xi'_1(0) \\ \xi_0(0) & \xi_1(0)\end{array}\!\right|^2,
\end{equation}
up to  a constant factor and higher order terms. So, when $u_0 \in \Sigma$, we have the following two cases:\\\\
\textbf{Case 5.} $u_0 \in \Sigma$, and $\xi_1$ is finite at the preimage $X$ of $u_0$. In this case, in view of Remark \ref{dersol}, the determinant in \eqref{asym} is the Wronskian of two independent solutions of  $\D_+\D_-\eta = u_0\eta$, so $\zeta(u) \sim t^2 = u - u_0$ and thus has a simple zero at $u_0$.\\ \\
\textbf{Case 6.} $u_0 \in \Sigma$, and $\xi_1$ has a pole of order $d$ at the preimage $X$ of $u_0$. In this case, renormalizing $\xi$ as in Case 2, we get that $\zeta$ has a pole of order $2d - 1$ at $u_0$.
\\\\
In the latter case, one can regard a pole of order $2d - 1$ as a pole of order $2d$ that collided with a simple zero. With this understanding, the number of zeros of $\zeta$ is equal to the number of branch points of $w$, while the number of poles of $\zeta$ is twice the number of poles of $\xi_1$ (with some zeros and poles possibly cancelling each other out). And since the number of zeros of $\zeta$ is equal to the number of its poles, it follows that the number of poles of $\xi_1$ is half the number of branch points of $w$, which is $2g + 2$. Furthermore, since $Z_+$ is not a pole of $\xi_1$, while $Z_-$ is its pole of order $1$ (see Proposition \ref{behinf}), it follows that the number of poles of $\xi_1$  in $\Gamma \setminus \{Z_\pm\} $ is exactly $g$, as desired. \end{proof}

\begin{corollary}\label{behfin}
In the rational case, all functions $\xi_k$ are holomorphic in $\Gamma\, \setminus \, \{Z_\pm\} $, while in the elliptic case all of them have at worst a simple pole at one and the same point $X_p$, and no other poles.
\end{corollary}
\begin{proof}
First note that $\xi_0 \equiv 1$ by construction. Furthermore, in the rational case $\xi_1$ is holomorphic in $\Gamma\, \setminus \, \{Z_\pm\} $  by Proposition \ref{gpoles}. So both $\xi_0$ and $\xi_1$ are holomorphic in that domain. At the same time, by~\eqref{xirel} we have
$$
\xi_{k+1} =\frac{1}{\alpha_k}( ( w - \beta_k)\xi_k -  \alpha_{k-1}\xi_{k-1}),
$$
and since $w$ is holomorphic in $\Gamma\, \setminus \, \{Z_\pm\} $ (see Table \ref{table}), it follows by induction that so are all $\xi_k$'s with $k \geq 0$. Analogously, using~\eqref{xirel} to express $\xi_{k-1}$, we  get that $\xi_k$'s with $k < 0$ are holomorphic in $\Gamma\, \setminus \, \{Z_\pm\} $ too. This proves the corollary in the rational case. \par
In the elliptic case, the argument is similar, but now  Proposition \ref{gpoles} implies that $\xi_1$ has a single pole in $\Gamma\, \setminus \, \{Z_\pm\} $. Denoting that pole by $X_p$, we get by induction that all $\xi_k$'s are holomorphic in $\Gamma\, \setminus \, \{Z_\pm, X_p\} $, as desired.
\end{proof}

\par
\section{Proof of Theorem \ref{thm2}: a self-dual polygon fixed by the pentagram map is Poncelet}
In this section we prove Theorem \ref{thm2}: any weakly convex self-dual twisted odd-gon $P$ projectively equivalent to its pentagram image $P'$ is Poncelet. To that end, we use the results of Section \ref{sec:sc} to obtain explicit formulas for coordinates of vertices of $P$ (see Section~\ref{sec:rat} for the case $g = 0$ and Section~\ref{ss:sd} for the case $g = 1$) and hence show that $P$ is a Poncelet polygon. 
%
%

\subsection{The rational case: degenerate Poncelet polygons}\label{sec:rat}
In this section, we prove Theorem \ref{thm2} in the case when the genus of $\Gamma$ is $0$, i.e. when $\Gamma$ is a rational curve. 
In that case, we will show that $P$ is a degenerate Poncelet polygon in the sense that the corresponding inscribed and circumscribed conics are not in general position. We keep the notation of the previous two sections. 
\begin{proposition}
The set $s^{-1}(1) := \{X \in \Gamma \mid s(X) = 1\}$  consists of either one or three points.
\end{proposition}
\begin{proof}
This set is invariant under the involution $\sigma$ and contains exactly one fixed point of that involution (see the proof of Proposition \ref{genus2}). So, it must contain odd number of points, and since $\deg s = 3$, it follows that $|s^{-1}(1)| = 1$ or $|s^{-1}(1)| = 3$.
\end{proof}
We consider the cases $|s^{-1}(1)| = 1$ and $|s^{-1}(1)| = 3$ separately. First, assume that $|s^{-1}(1)| = 3$. Denote points in $s^{-1}(1)$ by $A,B,C$, where $A$ and $B$ are switched by $\sigma$, while $C$ is fixed by~$\sigma$.
\begin{proposition}\label{abcind}
The vectors $\xi(A)$, $\xi(B)$, $\xi(C)$  form a basis of $\Ker(\D_+ -T^n \D_-)$. 
\end{proposition}
\begin{remark}
Note that the vectors $\xi(A)$, $\xi(B)$, $\xi(C)$ are finite because, by Corollary \ref{behfin}, the vector-function $\xi$ is holomorphic in $\Gamma \setminus \{Z_\pm\} $, and $A,B,C \neq Z_\pm$ since $s(Z_+) = \infty$ and $s(Z_-) = 0$ (see Table~\ref{table}).
\end{remark}
\begin{proof}[Proof of Proposition \ref{abcind}]
We have $\xi(A), \xi(B), \xi(C) \in \Ker(\D_+ -T^n \D_-)$ by \eqref{SEQN}, so it suffices to show that these vectors are linearly independent. To that end, recall that they are eigenvectors of the operator $\D_+\D_-$. Furthermore, the eigenvalue $w(C)$ corresponding to $\xi(C)$ is distinct from the eigenvalue $w(A) = w(B)$ corresponding to the other two vectors. So, it suffices to prove the independence of $\xi(A)$ and $\xi(B)$. But that follows from Lemma~\ref{indep}.
\end{proof}

Now recall that the polygon corresponding to the operator $\D_+ -T^n \D_-$ is $P$. So, by Proposition \ref{abcind} the vertices of $P$ (defined up to a projective transformation) are given by 
$
(\xi_k(A):\xi_k(B): \xi_k(C)) \in \P^2.
$
To explicitly compute the coordinates of vertices, we identify $\Gamma$ with $\bar \C$. Note that since automorphisms of $\bar \C$ act transitively on triples of points, the map $u \colon \Gamma \to \bar \C$ may be chosen in such a way that $u(Z_+ ) = 0$, $u(Z_-) = \infty$, and $u(C) = 1$. Then the involution $\sigma$, written in terms of $u$, is  $u \mapsto u^{-1}$, while the points $A,B$ are identified with $r$ and $r^{-1}$, where $r \in \C^*\setminus \{\pm 1\}$.  Furthermore, from Proposition \ref{behinf} and Corollary \ref{behfin} we get 
$
\xi_k(u) = c_k u^k,
$ where $c_k \neq 0$ is a constant. Therefore, the vertices of $P$ are given by
\begin{equation}\label{fund}
v_k = (r^k: r^{-k}: 1).
\end{equation}
So the polygon $P$ is inscribed in a conic with homogeneous equation \begin{equation}\label{circConic}x_1x_2 = x_3^2.\end{equation} Furthermore, since $P$ is self-dual, it is also circumscribed, and hence Poncelet. Thus, the proof of Theorem \ref{thm2} in the case when the spectral curve is rational and $|s^{-1}(1)| = 3$ is complete.

\begin{remark}A direct calculation shows that the conic inscribed in the polygon~\eqref{fund} is  \begin{equation}\label{insConic}x_1x_2 = \left(\frac{1}{2} + \frac{1}{4}(r + r^{-1})\right)x_3^2.\end{equation}
The conics \eqref{circConic} and \eqref{insConic} are tangent to each other at two points $(1:0:0)$ and $(0:1:0)$. In particular, they are not in general position (instead of four intersections we have two intersections of multiplicity~$2$). 
 \end{remark}

We now consider the case $|s^{-1}(1)| = 1$. 
To begin with, notice that this case can be thought of as a limit of the case  $|s^{-1}(1)| = 3$, with the points $A$, $B$, $C$ colliding together and forming a single point $D \in s^{-1}(1)$. This observation leads to the following version of Proposition \ref{abcind}: 
\begin{proposition}\label{dind}
The vectors $\xi(D)$, $\xi'(D)$, $\xi''(D)$  form a basis of $\Ker(\D_+ -T^n \D_-)$, where the derivatives are taken with respect to any local parameter near $D$.
\end{proposition}
\begin{proof}
First, notice that since the set $s^{-1}(1)$ consists of a single point $D$, the latter must be an order two branch point of the function $s$. In other words, we have $s'(D) = s''(D) = 0$. So, differentiating the equation $(\D_+ -sT^n \D_-)\xi = 0$ at the point $D$ twice, we get
$$
(\D_+ -T^n \D_-)\xi'(D) = (\D_+ -T^n \D_-)\xi''(D) = 0.
$$
Thus, we have $\xi(D), \xi'(D), \xi''(D) \in \Ker(\D_+ -T^n \D_-)$, and it suffices to show that these vectors are linearly independent. To that end, we differentiate the equation $(\D_+\D_- - w)\xi = 0$ twice at $D$. Using that $D$ is a branch point of $w$ and thus $w'(D) = 0$, we get
$$
(\D_+\D_- - w(D))\xi'(D) = 0, \quad (\D_+\D_- - w(D))\xi''(D) = w''(D)\xi(D).
$$
Furthermore, since the degree of the function $w$ is $2$, $D$ is an order $1$ branch point for $w$, so $w''(D) \neq 0$, which means that $\xi(D)$, $\xi'(D)$ are eigenvectors of $\D_+\D_-$, while $\xi''(D)$ is not. Furthermore, the vectors $\xi(D)$ and $\xi'(D)$ are linearly independent by Lemma \ref{indep2}. So, $\xi(D)$, $\xi'(D)$, $\xi''(D)$ are indeed independent, as desired.
\end{proof}
We now find the vertices of the polygon $P$ in the same fashion as in the case $|s^{-1}(1)| = 3$. Namely, choose an identification $u \colon \Gamma \to \P^1$ in such a way that $u(Z_+ ) = 0$, $u(Z_-) = \infty$, and $u(D) = 1$.  Then, as in the case  $|s^{-1}(1)| = 3$, we get $
\xi_k(u) = c_k u^k,
$ where $c_k \neq 0$ is a constant. In particular, at the point $D$ we get $\xi_k = c_k$, $\xi_k' = k c_k$, $\xi_k'' = k(k-1) c_k$, so up to a projective transformation the vertices of $P$ are given by
\begin{equation}\label{fundnilp}
v_k = (k: k^2: 1).
\end{equation}
These points belong to a conic \begin{equation}\label{circConic2}x_2x_3 = x_1^2,\end{equation}  so $P$ is inscribed and hence Poncelet. Thus, the proof of Theorem \ref{thm2} in the rational case is complete. \begin{remark}A direct calculation shows that the conic inscribed in the polygon \eqref{fundnilp} is \begin{equation}\label{insConic2}x_2x_3 = x_1^2 + \frac{1}{4}x_3^2.\end{equation} This is an even more degenerate case: the conics \eqref{circConic2} and \eqref{insConic2} intersect each other at one single point~$(0:1:0)$, of multiplicity $4$.

\end{remark}
%
\begin{remark}
Degenerate Poncelet polygons \eqref{fund} and \eqref{fundnilp} correspond to degenerations of an elliptic curve to Abelian groups $\C^*$ and $\C$ respectively.
Indeed, removing the tangency points with the inscribed conic  \eqref{insConic} from the circumscribed conic \eqref{circConic} we get an affine curve (a hyperbola) which is naturally isomorphic to $\C^*$. The vertices of the Poncelet polygon~\eqref{fund} are uniformly spaced on that hyperbola with respect to the $\C^*$ group structure. This should be compared with the case of genuine Poncelet polygons which are images of uniformly spaced points on an elliptic curve under a double covering map from the elliptic curve to the circumscribed conic. Thus, degenerate Poncelet polygons~\eqref{fund}  correspond to degenerations of an elliptic curve to $\C^*$. Likewise, removing from  \eqref{circConic2} the tangency point with~\eqref{insConic2} we get a parabola which identifies with $\C$, and points \eqref{fundnilp} are again uniformly spaced. Thus, degenerate Poncelet polygons~\eqref{fundnilp}  correspond to degenerations of an elliptic curve to $\C$.

\end{remark}

\subsection{The elliptic case: genuine Poncelet polygons}\label{ss:sd}
In this section, we complete the proof of Theorem \ref{thm2} in the case when the spectral curve $\Gamma$ has genus~$1$, i.e. is elliptic. The argument is similar to the rational case, but instead of elementary expressions~\eqref{fund} and \eqref{fundnilp}, we obtain formulas for vertices of $P$ in terms of theta functions. \par
Recall that in the elliptic case the involution $\sigma$ on $\Gamma$ has four fixed points, at three of which we have $s = 1$, while at the fourth one we have $s = -1$ (see the proof of Proposition \ref{genus2}). Denote those points by $A,B,C, D$, where $s(D) = -1$, and $s(A) = s(B) = s(C) = 1$.
\begin{proposition}
The directions in $\P^\infty$ determined by the values of the vector-function $\xi$ at the points $A$, $B$, $C$ (see Remark \ref{rem:holo}) are linearly independent.
\end{proposition}
\begin{proof}
They are eigendirections of the operator $\D_+\D_-$ corresponding to distinct eigenvalues $w(A)$, $w(B)$, $w(C)$.
\end{proof}
As in the rational case, it follows that the vectors $\xi(A)$, $\xi(B)$, $\xi(C)$ form a basis of $\Ker(\D_+ -T^n \D_-)$ (as usual, one may need to renormalize these vectors to ensure that they are finite, see Remark \ref{rem:holo}). Therefore, the vertices of the corresponding polygon $P$ (defined up to projective transformation) are given by 
$
(\xi_k(A):\xi_k(B): \xi_k(C)) \in \P^2.
$
\begin{remark}\label{psdclosed}
Since $z(A) = z(B) = z(C) = 1$ (see Remark \ref{z1}), it follows that the infinite vectors $\xi(A)$, $\xi(B)$, $\xi(C)$ are $n$-periodic, so the polygon $P$ is closed.
\end{remark}
To explicitly compute the coordinates of vertices of $P$, we identify $\Gamma$ with $\C \,/\, \Lambda$, where $\Lambda \subset \C$ is a lattice. Without loss of generality, assume that $\Lambda$ is spanned by $1$ and $\tau$, where $\tau$ is in the upper half-plane. 
Furthermore, one can choose an identification between $\Gamma$ and $\C \,/\, \Lambda$ in such a way that the point $D \in \Gamma$ gets identified with $d: = (1 + \tau)/2$. Then $\sigma$, understood as an involution in $\C\, /\, \Lambda$, is simply $u \mapsto -u$. So, $A,B,C$ must coincide with the remaining order $2$ points in $\C\, /\, \Lambda$, namely $0$, ${1}/{2}$, ${\tau}/{2}$. Without loss of generality, assume that $A = 0$, $B = {1}/{2}$, $C = {\tau}/{2}$.

We will express the vertices of $P$ using theta functions. Recall (see e.g. \cite{mumford2007tata}) that the \textit{theta function} corresponding to the lattice $\Lambda = \langle 1, \tau \rangle$ t is defined by
\begin{equation}
\theta(u) := \sum_{k \in \Z}\exp(\pi \mathrm{i}(2ku + k^2 \tau))
\end{equation}
where $\mathrm{i} = \sqrt{-1}$, and the dependence of $\theta$ on $\tau$ is suppressed for notational convenience. We will also need \textit{theta functions with (half-integer) characteristics}, defined by
\begin{align*}
\begin{gathered}
\theta_{00}(u) := \theta(u), \quad \theta_{01}(u) := \theta(u +  1/2), \quad \theta_{10}(u) := \exp(\pi \mathrm{i}(u+ \tau / 4)) \theta(u + \tau/2), \\
\theta_{11}(u) := \exp(\pi \mathrm{i}( u+{\tau}/{4}  +  {1}/{2})) \theta(u + (1+\tau)/2).
\end{gathered}
\end{align*}
Choose complex numbers $x_p, z_\pm \in \C$ whose images in $\C\, /\, \Lambda$ are the points $X_p, Z_\pm \in \Gamma$ (see Corollary \ref{behfin} for the definition of $X_p$). Note that the points $Z_\pm$ are interchanged by the involution $\sigma$, so $z_+ + z_- = 0$ modulo $\Lambda$. Therefore, without loss of generality one can assume that $z_+ + z_- = 1 + \tau$. Fix $z_\pm$ satisfying the latter condition, and let  $\delta:= z_+ - z_-$.

\begin{proposition}\label{prop:verttheta}
Up to a projective transformation, the vertices of the polygon $P_{}$ are given by
\begin{equation}\label{verttheta}
 v_k = ( \theta_{00}( k\delta + d - x_p) : \theta_{01}( k\delta + d - x_p) : \theta_{10}( k\delta + d - x_p)).
\end{equation}
\end{proposition}

Before proving this proposition, we recall standard properties of the theta function $\theta(u)$. It is easily seen from its definition that the theta function is holomorphic in $\C$, even, periodic with period $1$, and quasi-periodic with period $\tau$:
\begin{equation*}
\theta(-u) = \theta(u), \quad \theta(u+1) = \theta(u), \quad \theta(u+\tau) =\exp(-\pi \mathrm{i}(2u + \tau)) \theta(u).
\end{equation*}
In addition to that, one can show using the argument principle that the theta function has a unique simple zero at the point $d = (1 + \tau)/2$, and no other zeros in the fundamental parallelogram spanned by $1$ and $\tau$. These properties allow one  to express any meromorphic function on  $\C\, /\, \Lambda$ in terms of the theta function. The construction is based on the following well-known result: there exists a meromorphic function with zeros at $p_1, \dots, p_m \in \C \,/\, \Lambda$ and poles at $q_1, \dots, q_m \in \C \,/\, \Lambda$ if and only if $\sum p_k = \sum q_k$. So assume that we are given a collection of points with this property. Then the expression
\begin{equation}\label{thetaexpr}
f(u) := \frac{\prod\limits_{k=1}^m \theta(u - p_k + d)}{\prod\limits_{k=1}^m \theta(u - q_k +d)}
\end{equation}
defines a meromorphic function on $\C$ which can be easily seen to be periodic with respect to both $1$ and $\tau$ (here we regard $p_k$'s and $q_k$'s as points in $\C$ and assume that they are chosen in such a way that $\sum p_k = \sum q_k$ exactly, and not just modulo $\Lambda$). Therefore, this function can be viewed as a meromorphic function on  $\C \,/\, \Lambda$. Furthermore, the only zeros of $f(u)$ in  $\C \,/\, \Lambda$ are $p_k$'s, while its only poles are $q_k$'s. Since zeros and poles determine a meromorphic function up to a constant factor, it follows that any meromorphic function on  $\C \,/\, \Lambda$ with zeros at $p_1, \dots, p_m$ and poles $q_1, \dots, q_m$ can be written as \eqref{thetaexpr} times a constant. \par
\begin{proof}[Proof of Proposition \ref{prop:verttheta}]
Using Proposition~\ref{behinf} and Corollary \ref{behfin}, we get 
\begin{equation}\label{xiifla}
\xi_k(u) = c_k\frac{ \left(\theta(u - z_+ + d )\right)^k\theta(u - x_p + k\delta + d)}{\left(\theta(u - z_- + d)\right)^k\theta(u - x_p + d)},
\end{equation}
where $c_k$ is a non-zero constant, and the term containing $\delta$ is found by equating the sum of zeros with the sum of poles. Note that since we are only interested in the direction of the vector $\xi$, we may multiply all $\xi_k$'s by $\theta(u - x_p + d)$, which results in
$$
\tilde \xi_k(u) = c_k\frac{ \left(\theta(u - z_+ + d )\right)^k}{\left(\theta(u - z_- + d)\right)^k}{}{}\theta(u - x_p + k\delta + d).
$$
These are no longer meromorphic functions on $\C\, /\, \Lambda$, but still meromorphic functions on $\C$. Furthermore, in contrast to $\xi_k$'s, the functions $\tilde \xi_k$ are always finite at the points $0, {1}/{2}, {\tau}/{2} \in \C$ corresponding to $A,B,C \in \Gamma$, so the vertices of $P$ are given by
$
 (\tilde \xi_k(0):\tilde \xi_k({{1}/{2}}): \tilde \xi_k({{\tau}/{2}})).
$
Also notice that the values of the constants $c_k$ do not affect the latter expression, so one can assume that $c_k=1$. Under this assumption, we get
$$
\tilde \xi_k(0) =\frac{ \left(\theta( d - z_+ )\right)^k}{ \left(\theta( d - z_-)\right)^k}\theta(  k\delta + d - x_p) = \theta( k\delta + d - x_p),
$$
where the last equality follows from $d - z_- = -(d -z _+)$ and $\theta(-u) = \theta(u)$. Similarly, we have
\begin{align*}
\begin{gathered}
\tilde \xi_k\left({{1}/{2}}\right) =\frac{ \left(\theta({{1}/{2}} + d - z_+ )\right)^k}{ \left(\theta( {{1}/{2}} + d - z_-)\right)^k}\theta( {{1}/{2}} + k\delta + d - x_p) \\=  \frac{ \left(\theta(-{{1}/{2}} + d - z_+ )\right)^k}{ \left(\theta( {{1}/{2}} + d - z_-)\right)^k}\theta( {{1}/{2}} + k\delta + d - x_p)  = \theta( \textstyle{{1}/{2}} + k\delta + d - x_p).
\end{gathered}
\end{align*}
where  the second last equality follows from $1$-periodicity of $\theta$, and the last one from $d - z_- = -(d -z _+)$ and $\theta(-u) = \theta(u)$. 
Finally,
\begin{align*}
\begin{gathered}
\tilde \xi_k\left(\textstyle{{\tau}/{2}}\right) =  \frac{ \left(\theta({{\tau}/{2}} + d - z_+ )\right)^k}{ \left(\theta( {{\tau}/{2}} + d - z_-)\right)^k} \theta( {{\tau}/{2}} + k\delta + d - x_p) = \frac{ \left(\theta({{1}/{2}} + \tau - z_+ )\right)^k}{ \left(\theta( {{1}/{2}} + \tau - z_-)\right)^k}\theta( {{\tau}/{2}} + k\delta + d - x_p) \\
=  \exp( \pi k\mathrm{i} ( 2z_+ - 1- \tau)) \frac{ \left(\theta({{1}/{2}}  - z_+ )\right)^k}{\left(\theta( {{1}/{2}} + \tau - z_-)\right)^k}\theta( {{\tau}/{2}} + k\delta + d - x_p) \\  =  \exp( \pi k   \mathrm{i} \delta)\theta( \textstyle{{\tau}/{2}} + k\delta + d - x_p),
\end{gathered}
\end{align*}
where the second equality uses the definition  $d = (1 + \tau)/2$, the third one uses the formula for $\theta(u + \tau)$, while the last one uses that $\theta$ is even along with the relation $\delta = 2z_+ - 1- \tau$. 
\par
Now, to complete the proof it remains to rewrite the obtained fromulas using {theta functions with characteristics}.
We have  $\tilde \xi_k(0)= \theta_{00}( k\delta + d - x_p)$,  $\tilde \xi_k\left(\textstyle{{1}/{2}}\right)  = \theta_{01}( k\delta + d - x_p),$
while $\tilde \xi_k\left(\textstyle{{\tau}/{2}}\right)  = \theta_{10}( k\delta + d - x_p)$ up to a factor not depending on $k$. Since the latter factor does not affect the projective equivalence class of $P_{}$, one gets the desired formulas for vertices.
\end{proof}
\begin{remark}
Note that the functions $\xi_k$ may, but not necessarily do, have poles at $X_p$ (see Corollary~\ref{behfin}). However, formula~\eqref{xiifla} is valid anyway. Indeed, if $\xi_k$ does not have a pole at $X_p$, then its only pole is the point $Z_-$ (which is of order $k$), while its only zero is the point $Z_+$ (which is also of order $k$). So, we must have $kz_+ = kz_-$ modulo $\Lambda$, i.e. $k \delta \in \Lambda$. But then the factor ${\theta(u - x_p + k\delta + d)}/{\theta(u - x_p + d)}$ in \eqref{xiifla} is a non-vanishing holomorphic function, so the analytic properties (i.e. zeros and poles) of the right-hand side of \eqref{xiifla}  are the same as for the left-hand side, which means that these functions coincide for a suitable value of $c_k$.
\end{remark}
Now, to prove that $P_{}$ is Poncelet it suffices to establish the following: \begin{proposition} The image of the map $\Phi \colon \C \to \C\P^2$ given by
\begin{equation}\label{phiMap}
\Phi(u) := (\theta_{00}(u):\theta_{01}(u):\theta_{10}(u))
\end{equation}
is a conic.
\end{proposition}
\begin{proof} First of all, notice that the functions $\theta_{00}$, $\theta_{01}$, $\theta_{10}$ have no common zeros, so the mapping $\Phi$ is well-defined. Further, following \cite{mumford2007tata}, define the following operators $\mathcal S, \mathcal T$ on holomorphic functions on~$\C$:
$$
(\mathcal Sf)(u) := f(u + 1), \quad (\mathcal Tf)(u) = \exp(\pi \mathrm{i}( 2u+\tau)) f(u + \tau).
$$
Then \begin{equation}\label{Heis}
\mathcal S\theta_{jk} = (-1)^{j}\theta_{jk}, \quad \mathcal T\theta_{jk} = (-1)^{k}\theta_{jk}.
\end{equation}
In particular, we have $\mathcal S^2\theta_{jk} = \theta_{jk}$, $\mathcal T^2\theta_{jk} = \theta_{jk}$, which means that
\begin{equation}\label{charqp}
\theta_{jk}(u + 2) = \theta_{jk}(u), \quad \theta_{jk}(u + 2\tau) = \exp(-4\pi \mathrm{i}(u + \tau))\theta_{jk}(u).
\end{equation}
From the latter it follows that $\Phi$ descends to a holomorphic mapping $ \C\, /\, 2\Lambda \to \C\P^2$, so the image of $\Phi$ is an algebraic curve. To find the degree of that curve, one needs to find the number of its intersections with a generic line. Clearly, that number can be found as ${m}\,/\,{\deg \Phi}$, where $m$ is the number of zeros of a generic linear combination of $\theta_{00}$, $\theta_{01}$, $\theta_{10}$ in the fundamental parallelogram of the lattice $2L$, while $ \deg \Phi$ is the degree of $\Phi$, when the latter is regarded as a mapping  $ \C\, /\, 2\Lambda \to \C\P^2$. The number $m$ can be easily computed using quasi-periodicity relations \eqref{charqp} and the argument principle. That number is equal to $4$. Further, notice that the functions $\theta_{00}$, $\theta_{01}$, $\theta_{10}$ are even, so $\Phi(-x) = \Phi(x)$, which means that $\deg \Phi \geq 2$. Therefore, the degree of the image of $\Phi$ is either $2$ or $1$, i.e. the image of $\Phi$ is a conic or a straight line. However, it cannot be a straight line, because the functions $\theta_{k,j}$ are linearly independent by \eqref{Heis}. So, the image of $\Phi$ is a conic. 
\end{proof}
Thus, we conclude that the vertices \eqref{verttheta} of the polygon $P_{}$ lie on a conic. Since $P_{}$ is self-dual, it is also circumscribed about a conic, and hence Poncelet, q.e.d. So, Theorem \ref{thm2} is proved.
\begin{remark}
One can also explicitly describe the image of the mapping \eqref{phiMap} and hence the conic circumscribed about $P$ using Riemann's relation
\begin{equation}\label{rr0}
\sum_{j,k \in \{0,1\}} \theta_{jk}(\alpha_1)\theta_{jk}(\alpha_2)\theta_{jk}(\alpha_3)\theta_{jk}(\alpha_4) = 2\,\theta_{00}(\beta_1)\theta_{00}(\beta_2)\theta_{00}(\beta_3)\theta_{00}(\beta_4),
\end{equation}
where $\beta_1 := (\alpha_1 + \alpha_2 + \alpha_3 + \alpha_4)\,/\,2$,  $\beta_2 := (\alpha_1 + \alpha_2 - \alpha_3 - \alpha_4)\,/\,2$,  $\beta_3 := (\alpha_1 - \alpha_2 + \alpha_3 - \alpha_4)\,/\,2$,  $\beta_4 := (\alpha_1 - \alpha_2 - \alpha_3 + \alpha_4)\,/\,2$.
Taking $\alpha_1 = 0$, $\alpha_2 = u$, $\alpha_3 = v$, $\alpha_4 = u + v$, we get the identity
\begin{equation}\label{rr1}
\begin{gathered}
-\theta_{00}(0)\theta_{00}(u)\theta_{00}(v)\theta_{00}(u + v) +  \theta_{01}(0)\theta_{01}(u)\theta_{01}(v)\theta_{01}(u + v) \\+\,  \theta_{10}(0)\theta_{10}(u)\theta_{10}(v)\theta_{01}(u + v) =  0,
\end{gathered}
\end{equation}
which, after a further substitution $v = 0$, becomes
\begin{equation}\label{rr}
-  \theta_{00}^2(0) \theta_{00}^2(u) + \theta_{01}^2(0) \theta_{01}^2(u) + \theta_{10}^2(0) \theta_{10}^2(u) = 0.
\end{equation}
 So, the conic circumscribed about $P$ is given by
\begin{equation}\label{explconic}
-  \theta_{00}^2(0) x_1^2 + \theta_{01}^2(0) x_2^2 + \theta_{10}^2(0) x_3^2 = 0.
\end{equation}
Similarly, the conic inscribed in $P$ is
\begin{equation}\label{explconic2}
-  \theta_{00}^2(\delta/2) x_1^2 + \theta_{01}^2(\delta/2) x_2^2 + \theta_{10}^2(\delta/2) x_3^2 = 0.
\end{equation}
Indeed, let $t_k: =  k\delta + d - x_p$, $m:= k+ 1/2$, and  $t'_{m}:= (t_k + t_{k+1})/2$. Then, as follows from \eqref{rr}, the point
\begin{equation}\label{tangentpt}
v'_{m} := \left(\frac{\theta_{00}(0)\theta_{00}(t'_m)}{\theta_{00}(\delta/2)} : \frac{\theta_{01}(0)\theta_{01}(t'_m)}{\theta_{01}(\delta/2)} : \frac{\theta_{10}(0)\theta_{10}(t'_m)}{\theta_{10}(\delta/2)}\right)
\end{equation}
belongs to the conic  \eqref{explconic2}. Furthermore, the tangent line to  \eqref{explconic2} at $v'_{m}$ passes through the vertices $v_k$ and $v_{k+1}$ of $P$. Indeed, that is equivalent to the relation
\begin{align}
\begin{gathered}
- \theta_{00}(0)\theta_{00}(\delta/2)\theta_{00}(t'_m) \theta_{00}(t_{m \pm 1/2}) + \theta_{01}(0)\theta_{01}(\delta/2)\theta_{01}(t'_m)\theta_{01}(t_{m\pm 1/2}) \\ +\, \theta_{10}(0)\theta_{10}(\delta/2)\theta_{10}(t'_m)\theta_{10}(t_{m \pm 1/2}) = 0,
\end{gathered}
\end{align}
which is a particular case of \eqref{rr1} corresponding to $u = t_{m \pm 1/2} $, $v = \mp \delta/2$. So indeed the polygon $P$ is circumscribed about the conic~\eqref{explconic2}.\par
\end{remark}
\begin{remark}
Note that formula \eqref{verttheta} describes a \textit{family} of polygons, parametrized by $x_p$. Our argument shows that all these polygons are inscribed in one and the same conic~\eqref{explconic} and circumscribed about one and the same conic \eqref{explconic2}. So, polygons~\eqref{verttheta} form what is called a \textit{Poncelet family}, i.e. a family of polygons inscribed in the same conic and circumscribed about the same conic (recall that every Poncelet polygon is a member of such a family by Poncelet's porism). Also note that the expression~\eqref{verttheta} is periodic in $x_p$ with the periods given by the lattice $2\Lambda$. So, the Poncelet family containing our polygon $P$ is parametrized by the elliptic curve $\C\, /\, 2\Lambda$, which is a $4$-to-$1$ covering of the spectral curve $\Gamma = \C \,/\, \Lambda$. As a corollary, the Poncelet family containing $P$ contains four polygons projectively equivalent to $P$: one of those polygons is $P$, while the other three can be obtained from $P$ by replacing $x_p$ in formula~\eqref{verttheta} with $x_p + 1$, $x_p + \tau$, and $x_p + 1 + \tau$. This quadruple of polygons admits a geometric description when the circumscribed conic $C_1$ and inscribed one $C_2$ are confocal. In this case, these polygons can be obtained from $P$ by means of reflection with respect to the common symmetry axes of $C_1$, $C_2$.\par
This argument also shows that the spectral curve is the same for all polygons in a Poncelet family. Using a different approach, this was earlier proved in \cite{schwartz2015pentagram}. Formulas for Poncelet families similar to~\eqref{verttheta} are given in \cite{veselov1988integrable}.
\end{remark}
\begin{remark}
Note that since the polygon $P_{}$ is closed (Remark \ref{psdclosed}), the expression \eqref{verttheta} must be $n$-periodic in $k$. Therefore, we must have $n\delta \in 2\Lambda$. Another way to see this is to consider the function $(s-1)\mu_-$ on $\Gamma$. Using Table \ref{table} and the fact that $s(A) = s(B) = s(C) = 1$, we conclude that this function has simple zeros at $A$, $B$, $C$, a zero of order $(n-3)/2$ at $Z_-$, and a pole of order $(n+3)/2$ at $Z_+$. So, we have
$
0 + {1}/{2} + {\tau}/{2}+ {(n-3)}/{2} \cdot z_- =  {(n+3)/}{2} \cdot z_+ \,(\mathrm{mod}\,\Lambda),
$
 which implies
$
 {n}/{2} \cdot \delta =  {n}/{2} \cdot(z_+ - z_-) =  {1}/{2} + {\tau}/{2} - {3}/{2} \cdot (z_+ + z_-) = -2d = 0  \,(\mathrm{mod}\,\Lambda),
$
and thus $n\delta \in 2\Lambda$, as desired.\par
Also note that formula \eqref{verttheta} still defines a Poncelet polygon if $n\delta \in \Lambda  \setminus 2\Lambda$. It is then a \textit{twisted} $n$-gon, which can also be viewed as a closed $2n$-gon. Such twisted Poncelet polygons do not arise in our setting, because they are not fixed points of the pentagram map.

\end{remark}

\section{Proof of Theorem \ref{thm1}: a closed polygon fixed by the pentagram map is Poncelet}

In this section, we derive Theorem \ref{thm1} from Theorem \ref{thm2}. To that end, we first show, in Section \ref{sec:sd}, that the self-duality assumption of Theorem \ref{thm2} is not very restrictive. Namely, any polygon satisfying all the assumptions of the theorem except for possibly self-duality, can be transformed, by means of rescaling~\eqref{rescaling} with $s > 0$, into a self-dual polygon. From that we conclude that a polygon as in Theorem~\ref{thm1} (i.e. weakly convex, closed, and projectively equivalent to its pentagram image) must be Poncelet up to rescaling \eqref{rescaling} with $s > 0$. So to show that that polygon is actually Poncelet we need to prove that the rescaling is trivial, i.e. corresponds to $s = 1$. To that end, we show that if a weakly convex Poncelet polygon is rescaled in a non-trivial way, then the resulting polygon cannot be closed. This is done separately in the rational (see Section \ref{ss:genrat}) and elliptic (see Section \ref{ss:genell}) cases. In the rational case we have a very simple explicit description of the corresponding degenerate Poncelet polygons (see Section~\ref{sec:rat}), so in that case the proof is completely elementary. As for the the elliptic situation, in that case the proof relies on the study of the real part of the corresponding elliptic curve and location of various special points within that real part.  
\par
\subsection{Self-duality up to rescaling}\label{sec:sd}
\begin{proposition}\label{prop:sd}
Assume that a closed or twisted weakly convex polygon $P$ is projectively equivalent to its pentagram image  $P'$. Then one can choose the $n$-periodic operator $\D$ of the form \eqref{diffOp2} associated with $P$ in such a way that the corresponding commuting operators $\Dl , \Dr $ given by \eqref{dldr} satisfy
\begin{equation}\label{sde}
 \Dr  = - s_0  T^{n}  \Dl ^*
\end{equation}
for certain $ s_0  \in \R_+$.
\end{proposition}

\begin{proof}
Let $ \D$ be an $n$-periodic operator corresponding to $P$ such that the corresponding operators
$
 \Dl ,   \Dr 
$
commute, and, moreover, the coefficients of $ \D$ satisfy the alternating signs condition \eqref{altCond} (such $ \D$ exists by Proposition \ref{prop:cdo}). Then the operator $T^{-n} \Dl  \Dr$ has the form
\begin{equation}\label{symmeOP}
T^{-n} \Dl  \Dr = \alpha T^{-1} + \beta + \gamma T.
\end{equation}
Moreover from the alternating signs condition we have $\alpha_k, \gamma_k > 0$ for all $k \in \Z$. Therefore, the operator~\eqref{symmeOP} can be symmetrized. Namely, there exists a positive quasi-periodic sequence $\lambda$ such that the operator
$
\lambda T^{-n} \Dl  \Dr  \lambda^{-1}
$ is self-dual. That sequence can be found from the equation
$
{\lambda_{k+1}}/{\lambda_k} = \sqrt{{\gamma_k}/{\alpha_{k+1}}}
$. So, conjugating $\D$ by $\lambda$ if needed, we may assume that the operator \eqref{symmeOP} is self-dual, meaning that
\begin{equation}\label{prodsd}
T^{-n} \Dl  \Dr  = T^{n}  \Dr ^* \Dl ^*.
\end{equation}
We now show that under that assumption we must have \eqref{sde}. Let $z_l$, $z_r$ be the monodromies of $\D_l$, $\D_r$ respectively. Then, by the second statement of Proposition \ref{alt}, we have $0 < z_l < z_r$. Furthermore, since $\Dl$ and $\Dr$ commute, it follows that the kernels of both of them are contained in $\Ker \Dl\Dr$. So, the spectrum of the monodromy of  $\Dl\Dr$ is $\{z_l, z_r\}$. Moreover, we have
$$
\Ker (\Dl\Dr)\vert_{\v{z_l}} = \Ker \Dl, \quad \Ker (\Dl\Dr)\vert_{\v{z_r}} = \Ker \Dr.
$$

Similarly, using that the monodromy of $\Dl^*$ and $\Dr^*$ is given by $z_l^{-1}$ and $z_r^{-1}$ respectively, we conclude that the monodromy of $\Dr ^* \Dl ^*$ is $\{z_l^{-1}, z_r^{-1}\}$, which, in view of \eqref{prodsd} and the inequality $0 < z_l < z_r$ implies $z_l = z_r^{-1}$. Furthermore, we have
$$
\Ker \Dl^* = \Ker (\Dr^*\Dl^*)\vert_{\v{z_l^{-1}}}=  \Ker (\Dl\Dr)\vert_{\v{z_l^{-1}}} = \Ker (\Dl\Dr)\vert_{\v{z_r}} = \Ker \Dr,
$$
so
\begin{equation}\label{asde}
\Dl ^* = T^{-n} \mu \Dr 
\end{equation}
for a certain $n$-periodic sequence $\mu$ of non-zero real numbers. 
Taking the dual of both sides, we also get
$
\Dr ^* = T^{-n}\Dl  \mu^{-1},
$
so
$$
\Dl ^* \Dr ^* = T^{-2n}\mu \Dr  \Dl  \mu^{-1} =  T^{-2n}\mu \Dl \Dr  \mu^{-1}.
$$
At the same time, we have
$$
\Dl ^* \Dr ^* = \Dr ^* \Dl ^* = T^{-2n} \Dl  \Dr ,
$$
so $\mu$ commutes with $ \Dl  \Dr $. But that is only possible if $\mu$ is a constant sequence $\mu_k = c$. So,~\eqref{asde} implies~\eqref{sde}, with $ s_0  = -c^{-1}$. Furthermore, since the coefficient of the highest degree term in $\D_r$ is a sequence of negative numbers, while the coefficient of  the coefficient of the highest degree term in $\D_l^*$ is a sequence of positive numbers, equation~\eqref{sde} can only be satisfied for $ s_0  > 0$, as desired.
\end{proof}
\begin{corollary}\label{cor:rescalingPon}
Assume that a closed or twisted weakly convex polygon $P$ is projectively equivalent to its pentagram image $P'$. Then there exists a polygon $P_{sd}$ with the same properties which is, in addition, self-dual (and hence Poncelet by Theorem \ref{thm2}), such that $P = R_{s_0} (P_{sd})$ where $R_{s_0} $ is the rescaling \eqref{rescaling} with $ s = s_0  > 0$.
\end{corollary}
\begin{proof}
Take the operator $\D$ provided by Proposition {\ref{prop:sd}}. It has the form
$
\D = \Dl  -   s_0  T^n \Dl ^*,
$
where $ s_0  \in \R_+$. Consider also the operator
$
 \D_{sd} = \Dl  -  T^n \Dl ^*,
$
and the associated polygon $P_{sd}$.
 Then, by Corollary \ref{cor:rescalingDO}, we have $P = R_{s_0}  (P_{sd})$. In particular, $P_{sd}$ is projectively equivalent to its pentagram image (because the pentagram map commutes with rescaling) and weakly convex (by the third statement of Proposition~\ref{alt}). Furthermore, we have $
 \D_{sd}^* =  -T^{-n}  \D_{sd},
$
so $P_{sd}$  is self-dual, as desired.
\end{proof}


\par
\subsection{End of proof in the rational case}\label{ss:genrat}
Let $P$ be a weakly convex closed polygon projectively equivalent to its pentagram image $P'$, as in Theorem \ref{thm1}.  Then, by Corollary \ref{cor:rescalingPon}, there exists a generally speaking twisted polygon $P_{sd}$ such that $P = R_{s_0} (P_{sd})$ for some $s_0 > 0$, and $P_{sd}$ is self-dual. Consider the spectral curve associated with $P_{sd}$, constructed in the proof of Theorem \ref{thm2}. In this section, we prove Theorem \ref{thm1} in the case when the genus of $\Gamma$ is $0$, i.e. when $\Gamma$ is rational. To that end, we will show that $s_0 = 1$, so $P = P_{sd}$ and hence Poncelet.


\par

As we know from Section \ref{sec:rat}, in the rational case the vertices of $P_{sd}$ are given by \eqref{fund} or \eqref{fundnilp}. In case \eqref{fund}, the associated difference operator reads
\begin{equation}\label{ccoperator}
  \D_{sd} =  T^{{(n-3)}/{2}}-aT^{{(n-1)}/{2}} + aT^{{(n+1)}/{2}} -T^{{(n+3)}/{2}},
\end{equation}
where $a$ is such that the roots of the corresponding characteristic polynomial $1 - ax + ax^2 - x^3$ are $r, r^{-1}$, and $1$ (note that since the polygon $P_{sd}$ is real, $a$ must be real too, so we must have $|r| = 1$). Indeed, the kernel of such an operator is spanned by the sequences $r^k$, $r^{-k}$, and a constant sequence, so the associated polygon is precisely \eqref{fund}. Likewise, in the case  \eqref{fundnilp}, the associated difference operator is also of the form \eqref{ccoperator}, with $a = 3$. So, since the polygon $P_{sd}$ is defined by the operator \eqref{ccoperator}, the polygon $P = R_{s_0}(P_{sd})$ is defined by
%
$$
\D  = T^{{(n-3)}/{2}}-aT^{{(n-1)}/{2}} + s_0(aT^{{(n+1)}/{2}} -T^{{(n+3)}/{2}}).
$$
 The kernel of this operator is spanned by the sequences $x_1^k$, $x_2^k$, $x_3^k$, where $x_1$, $x_2$, $x_3$ are the roots of the characteristic polynomial $h(x) := 1 - ax + s_0(ax^2 - x^3)$ (note that we do not need to consider the case of multiple roots, because in that case the monodromy of $\D$ is not diagonalizable, and the polygon $P$ cannot be closed). Moreover, since $P$ is closed, we must have $x_1^n = x_2 ^n  = x_3^n$, so $|x_1| = |x_2| = |x_3| = \lambda$, where $\lambda > 0$ is a real number. So, the roots of the polynomial   $h(\lambda x) = 1 - a\lambda x + s_0(a\lambda ^2x^2 - \lambda^3x^3)$ must all have absolute value $1$. Also taking into account that this polynomial is real, and that $s_0 \lambda^3 > 0$, we conclude that the roots of $h(\lambda x)$ are of the form $1, \alpha, \bar \alpha$, where $|\alpha| = 1$. But this yields $s_0 \lambda^3 = 1$ and $s_0 \lambda^2 = \lambda$, so $s_0 = 1$. Therefore, the polygon $P$ coincides with $P_{sd}$ and hence Poncelet. Thus, the proof of Theorem \ref{thm1} in the rational case is complete.
 \begin{remark}
One can also give a more concrete description of $P$, as follows. Since the vertices of $P$ are given by \eqref{fund} (with \eqref{fundnilp} being impossible due to closedness of $P$), and $P$ is a closed $n$-gon, it follows that $r^n = 1$. So, applying a linear transformation to \eqref{fund}, we get a polygon whose vertices have affine coordinates
$
\cos({2\pi mk}/{n})$, $\sin({2\pi mk}/{n}),$ where ${2\pi m}/{n} = \arg r$. In particular, if $m = 1$, then $P$ is a regular $n$-gon.
\end{remark}

\par
\subsection{End of proof in the elliptic case}\label{ss:genell}
In this section, we prove Theorem \ref{thm1} in the case when the genus of $\Gamma$ is $1$, i.e. when $\Gamma$ is elliptic. As in the rational case, we show that $s_0 = 1$, so $P = P_{sd}$ and hence Poncelet. We keep the notation of Sections \ref{ss:genus} and \ref{ss:sd}\par

Recall that a \textit{real structure} on a Riemann surface $\Gamma$ is an anti-holomorphic involution $\rho \colon \Gamma \to \Gamma$. The \textit{real part} $\Gamma_\R$ of $\Gamma$ (with respect to the real structure $\rho$) is then defined as the set of fixed points of $\rho$: $\Gamma_\R := \{ X \in \Gamma \mid \rho(X) = X\}$. A meromorphic function $f$ on $\Gamma$ is called a \textit{real function} if $\rho^*f = \bar f$. Real functions take real values at real points (i.e. points in $\Gamma_\R$).\par
In our case, the spectral curve $\Gamma$ is endowed with a {real structure} $\rho \colon \Gamma \to \Gamma$ induced by the involution $(z,w) \mapsto (\bar z, \bar w)$ on the affine spectral curve $\Gamma_a$. 

\begin{proposition}\label{realfunctions}
The functions $z$, $w$, $\mu_\pm$, $s$, $\xi$ on $\Gamma$ are real (see Section \ref{ss:genus} for the definition of those functions).
\end{proposition}
\begin{proof}
The functions $z,w$ are real by construction of the real structure $\rho$. To prove that the vector-function $\xi$ is real, notice that it is defined by equations \eqref{xidefeqns} up to a scalar factor. Taking the complex conjugate of those equations and then applying $\rho^*$, we get that $\rho^* \bar \xi = f\xi$ for a certain meromorphic function $f$. But then the normalization condition $\xi_1 = 1$ implies $f = 1$, as desired. Now, the reality of the functions $\mu_\pm$ follows from equation \eqref{mupluschar}, while reality of $s$ follows from its definition \eqref{SDEF}.
\end{proof}
\begin{corollary}\label{realpoints}
The points $Z_\pm,S_\pm,A,B,C,D \in \Gamma$ are real  (see Section \ref{ss:genus} for the definition of $Z_\pm, S_\pm$ and Section \ref{ss:sd} for the definition of $A,B,C,D$). 
\end{corollary}
\begin{proof}
Since $z$ is real function (Proposition \ref{realfunctions}), the involution $\rho$ takes zeros of $z$ to zeros of $z$. But the only zero of $z$ is $Z_+$ (see Table \ref{table}), so the latter must be real. Analogously, $Z_-$ is real as the only pole of $z$, $S_+$ is real as the only simple zero of $s$, $S_-$ is real as the only simple pole of $s$, while $D$ is real as the only point where both $s$ and $z$ are equal to $-1$ (see Remark \ref{z1}). To show that $A, B, C$ are real, observe that they constitute the set of points where $s = 1$, so $\rho$ takes the set $\{A, B, C\}$ to itself. Further, notice that the values of the function $w$ at $A, B, C$ are eigenvalues of a self-adjoint operator $(\D_+\D_-)\vert_{\v{1}}$ and hence real. Furthermore, those values are distinct, because $A,B,C$ are branch points of $w$, while $\deg w = 2$. But if, say, $\rho(A) = B$, then we must have $w(B) = \bar w(A)$, which is not possible since $w(A), w(B)$ are real and distinct. So, $\rho$ cannot permute $\{A, B, C\}$ and thus preserves each of them.
\end{proof}
\begin{corollary}
The real part $\Gamma_\R$ of $\Gamma$ consists of two disjoint circles.
\end{corollary}
\begin{proof}
The real part of any Riemann surface consists a finite number of disjoint circles (ovals). Furthermore, since the genus of $\Gamma$ is $1$, the number of connected components of $\Gamma_\R$ is at most $2$ by Harnack's theorem. At the same time, the number of connected components is non-zero since the real part $\Gamma_\R$ of $\Gamma$ is not empty (by Corollary \ref{realpoints}). So, it remains to determine whether the number of connected components is $1$ or $2$. These cases can be distinguished by counting the number of real points of order $2$ on $\Gamma$. Namely, if $\Gamma$ is identified with $\C \,/\, \Lambda$ in such a way that $0$ is a real point, then $\Gamma_\R$ is a subgroup of $\Gamma$ isomorphic to $S^1$ if $\Gamma_\R$ is connected, and $S^1 \times \Z_2$ if $\Gamma_\R$ has two components. So the number of real order $2$ points in $\Gamma_\R$ is $2^{m}$, where $m$ is the number of components of $\Gamma$. Identifying $\Gamma$ with $\C / L$ as in Section \ref{ss:sd}, we see that the order $2$ points are $A,B,C,D$, which are all real. So, $m=2$, q.e.d.
\end{proof}
This argument also shows that one of the components of $\Gamma_\R$ contains the point $D$ and one of the points $\{A, B, C\}$, while the second component of $\Gamma_\R$ contains the remaining two points. Without loss of generality, assume that $C$ and $D$ are located in the same component. Denote that component by~$\Gamma_\R^0$. 
\begin{proposition}
We have $Z_\pm, S_\pm \in \Gamma_\R^0$.
\end{proposition}
\begin{proof}
The function $z$ is real-valued on $ \Gamma_\R^0$ and satisfies $z(C) = 1$, $z(D) = -1$. So, there must be at least two points on $ \Gamma_\R^0$ where $z$ changes sign. But the only points which have this property are $Z_\pm$ (see Table \ref{table}).
Similarly, $s(C) = 1$, $s(D) = -1$, so the function $s$ should also change sign at two points. Moreover, these cannot be the points $Z_\pm$, since at those points $s$ has a zero and a pole of order $2$. So, we must have $S_\pm \in \Gamma_\R^0$, as desired.
\end{proof}

\begin{figure}[t]
\centering
\begin{tikzpicture}[scale = 1]

\draw (0,0) circle (1);
\coordinate (C) at (0,1);
\coordinate (D) at (0,-1);
\coordinate (Zp) at (-0.8,0.6);
\coordinate (Zm) at (0.8,0.6);
\coordinate (Sp) at (-0.8,-0.6);
\coordinate (Sm) at (0.8,-0.6);

\node[label={[shift={(0,-0.05)}]above:${C}$}] at (C) () {};
\node[label={[shift={(0,0.05)}]below:${D}$}] at (D) () {};
\node[label={[shift={(0,0)}]left:${S_+}$}] at (Zp) () {};
\node[label={[shift={(0,0)}]right:${S_-}$}] at (Zm) () {};
\node[label={[shift={(0,0)}]left:${Z_+}$}] at (Sp) () {};
\node[label={[shift={(0,0)}]right:${Z_-}$}] at (Sm) () {};
\fill (C) circle (0.07);
\fill (D) circle (0.07);
\fill (Zp) circle (0.07);
\fill (Zm) circle (0.07);
\fill (Sp) circle (0.07);
\fill (Sm) circle (0.07);

\end{tikzpicture}
\caption{Location of the points $C,D,Z_\pm,S_\pm$ in the component  $\Gamma_\R^0$ of the real part of the spectral curve.}\label{Fig:realpart}
\end{figure}
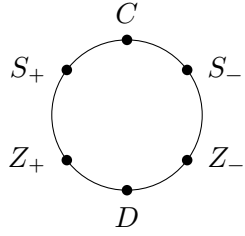

\begin{proposition}\label{cyclicOrder}
The cyclic order of the points  $C,D,Z_\pm,S_\pm$ on $ \Gamma_\R^0$ is as shown in Figure \ref{Fig:realpart}.
\end{proposition}
The proof is based on the following two lemmas.
\begin{lemma}\label{flemma2}
We have $z(S_+) \in (0,1)$.
\end{lemma}
\begin{proof}
Without loss of generality, assume that the vector $\xi(S_+)$ is finite and non-zero (see Remark~\ref{rem:holo}). Then, using the definition of the function $\mu_+$ and the fact that $\mu_+(S_+) = 0$ (see Table \ref{table}), we get
$
\D_+ \xi(S_+) =\mu_+(S_+) \xi(S_+) = 0.
$ Therefore, $ \xi(S_+) $ spans the kernel of the operator $\D_+$, while $z(S_+)$ is the monodromy of that operator. So, by the second statement of Proposition \ref{alt}, the number $z(S_+)$ is positive and is less than the monodromy of $\D_r = -T^n\D_+^*$ But the monodromy of the latter operator is the same as the monodromy of $\D_+^*$, which is $z(S_+)^{-1}$. So, we get
$
0 < z(S_+) < z(S_+)^{-1},
$
and the result follows.
\end{proof}
\begin{lemma}\label{flemma1}
The only point in $\Gamma_\R^0$ where $z =  1$ is the point $C$.
\end{lemma}
\begin{proof}
Assume that $X \in \Gamma_\R^0$ and $z(X)  =  1$. Then the latter condition in particular implies $X \neq Z_\pm$. Therefore, without loss of generality, we may assume that the vector $\xi(X)$ is finite and non-zero (if not, we renormalize $\xi$, see Remark \ref{rem:holo}). Under this assumption, using the inner product \eqref{pairing} on $\v{\pm 1}$, we get
$$
\mu_+(X) \left\langle  \xi(X), \xi(X) \right\rangle =  \left\langle  \D_+\xi(X), \xi(X) \right\rangle = \left\langle  \xi(X), \D_-\xi(X) \right\rangle = \mu_-(X) \left\langle  \xi(X), \xi(X) \right\rangle.
$$
Furthermore, since the vector $\xi(X)$ is real, it follows that $\left\langle  \xi(X), \xi(X) \right\rangle > 0$, and thus $\mu_+(X) = \mu_-(X)$. So, using formula \eqref{SDEF} for the function $s$, we get $s(X) = z(X)^{-1} =  1$ (here we use that the value $\mu_+(X) = \mu_-(X)$ is finite and non-zero, which is true because the functions $\mu_\pm$ do not have common zeros or poles, see Table \ref{table}). Furthermore, recall that  the set of points where $s = 1$ consists of the point $C$, plus points $A$ and $B$ which do not belong to $\Gamma_\R^0$. The result follows.
\end{proof}

\begin{proof}[Proof of Proposition \ref{cyclicOrder}]
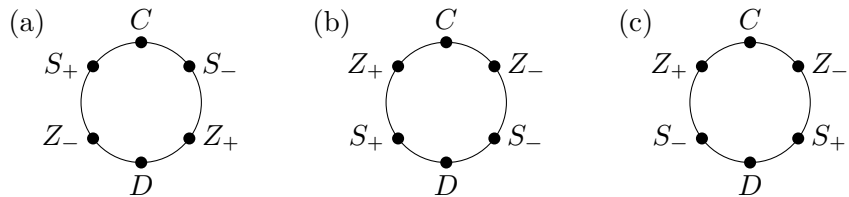
\begin{figure}[t]
\centering
\begin{tikzpicture}[scale = 0.9]

\node at (0,0) () {

\begin{tikzpicture}[scale = 0.8]
\node at (-1.9,1.3) () {(a)};
\draw (0,0) circle (1);
\coordinate (C) at (0,1);
\coordinate (D) at (0,-1);
\coordinate (Zp) at (-0.8,0.6);
\coordinate (Zm) at (0.8,0.6);
\coordinate (Sp) at (-0.8,-0.6);
\coordinate (Sm) at (0.8,-0.6);

\node[label={[shift={(0,-0.1)}]above:${C}$}] at (C) () {};
\node[label={[shift={(0,0.1)}]below:${D}$}] at (D) () {};
\node[label={[shift={(0.1,0)}]left:${S_+}$}] at (Zp) () {};
\node[label={[shift={(-0.1,0)}]right:${S_-}$}] at (Zm) () {};
\node[label={[shift={(0.1,0)}]left:${Z_-}$}] at (Sp) () {};
\node[label={[shift={(-0.1,0)}]right:${Z_+}$}] at (Sm) () {};
\fill (C) circle (0.1);
\fill (D) circle (0.1);
\fill (Zp) circle (0.1);
\fill (Zm) circle (0.1);
\fill (Sp) circle (0.1);
\fill (Sm) circle (0.1);

\end{tikzpicture}
};
\node at (4.5,0) () {
\begin{tikzpicture}[scale = 0.8]
\node at (-1.9,1.3) () {(b)};
\draw (0,0) circle (1);
\coordinate (C) at (0,1);
\coordinate (D) at (0,-1);
\coordinate (Zp) at (-0.8,0.6);
\coordinate (Zm) at (0.8,0.6);
\coordinate (Sp) at (-0.8,-0.6);
\coordinate (Sm) at (0.8,-0.6);

\node[label={[shift={(0,-0.1)}]above:${C}$}] at (C) () {};
\node[label={[shift={(0,0.1)}]below:${D}$}] at (D) () {};
\node[label={[shift={(0.1,0)}]left:${Z_+}$}] at (Zp) () {};
\node[label={[shift={(-0.1,0)}]right:${Z_-}$}] at (Zm) () {};
\node[label={[shift={(0.1,0)}]left:${S_+}$}] at (Sp) () {};
\node[label={[shift={(-0.1,0)}]right:${S_-}$}] at (Sm) () {};
\fill (C) circle (0.1);
\fill (D) circle (0.1);
\fill (Zp) circle (0.1);
\fill (Zm) circle (0.1);
\fill (Sp) circle (0.1);
\fill (Sm) circle (0.1);
\end{tikzpicture}
};
\node at (9,0) () {
\begin{tikzpicture}[scale = 0.8]
\node at (-1.9,1.3) () {(c)};
\draw (0,0) circle (1);
\coordinate (C) at (0,1);
\coordinate (D) at (0,-1);
\coordinate (Zp) at (-0.8,0.6);
\coordinate (Zm) at (0.8,0.6);
\coordinate (Sp) at (-0.8,-0.6);
\coordinate (Sm) at (0.8,-0.6);

\node[label={[shift={(0,-0.1)}]above:${C}$}] at (C) () {};
\node[label={[shift={(0,0.1)}]below:${D}$}] at (D) () {};
\node[label={[shift={(0.1,0)}]left:${Z_+}$}] at (Zp) () {};
\node[label={[shift={(-0.1,0)}]right:${Z_-}$}] at (Zm) () {};
\node[label={[shift={(0.1,0)}]left:${S_-}$}] at (Sp) () {};
\node[label={[shift={(-0.1,0)}]right:${S_+}$}] at (Sm) () {};
\fill (C) circle (0.1);
\fill (D) circle (0.1);
\fill (Zp) circle (0.1);
\fill (Zm) circle (0.1);
\fill (Sp) circle (0.1);
\fill (Sm) circle (0.1);
\end{tikzpicture}
};

\end{tikzpicture}
\caption{Impossible locations of the points $C,D,Z_\pm,S_\pm$ on  $\Gamma_\R^0$.}\label{Fig:realpartimp}
\end{figure}
The restriction of the involution $\sigma$ to $\Gamma_\R^0$ preserves the points $C,D$, interchanges $Z_+$ with $Z_-$, and interchanges $S_+$ with $S_-$. For this reason, the only possible locations of those points on $\Gamma_\R^0$ are the one depicted in Figure \ref{Fig:realpart}, as well as the ones depicted in Figure \ref{Fig:realpartimp}. Assume that $C,D,Z_\pm,S_\pm$ are located as in Figure \ref{Fig:realpartimp}a. Then, since $z(S_+) \in (0,1)$ by Lemma \ref{flemma2}, while $Z_-$ is a pole of $z$, there must be a point $X$ in the open arc  $(S_+, Z_-)$ such that $z(X) = 1$ or $z(X) = 0$ (here and below $(X,Y)$ denotes an open arc going from $X$ to $Y$ in counter-clockwise direction). However, the former is impossible by Lemma \ref{flemma1}, while the latter is impossible since the only zero of $z$ is the point $Z_+$. So, the points cannot be located as in Figure \ref{Fig:realpartimp}a. Further, since $z(D) = -1$, while the only points where $z$ changes sign are $Z_\pm$, in Figures \ref{Fig:realpartimp}b and  \ref{Fig:realpartimp}c we must have $z(S_+) < 0$, which is impossible by Lemma \ref{flemma2}. So, the points $C,D,Z_\pm,S_\pm$ are located as in Figure~\ref{Fig:realpart}.
\end{proof}
Now recall that the elliptic curve $\Gamma$ is associated with a Poncelet $n$-gon $P_{sd}$, and in addition we have a closed $n$-gon $ P = R_{s_0}(P_{sd})$, where $s_0 > 0$. Our aim is to show that $s_0 = 1$.
\begin{proposition}
There is a point $X_0 \in \Gamma_\R^0$ such that $s(X_0) = s_0$ and $z(X_0) = s_0^{-n/3}$.
\end{proposition}
\begin{proof}
The function $s$ has one simple pole and one double pole in $\Gamma_\R^0$ (see Table \ref{table}). Therefore, the degree of the mapping $s \colon \Gamma_\R^0 \to \RP^1$ is equal to $\pm 1$ (depending on the orientations). In particular, this mapping is surjective. So there exists $X_0 \in \Gamma_\R^0$ such that $s(X_0) = s_0$. To show that $z(X_0) = s_0^{-n/3}$, recall that the polygon $P$ associated with the operator $\D_+ -s_0T^n \D_-$ is closed. Therefore, the monodromy of the that operator has the form $\lambda \Id$. At the same time, since $\D_- = \D_+^*$, the explicit form of that operator is
$$\D_+ -s_0T^n \D_- = aT^{(n-3)/2} + bT^{(n-1)/2} - s_0 \tilde b T^{(n+1)/2} - s_0 \tilde aT^{(n+3)/2},$$ where the sequences $\tilde a$, $\tilde b$ coincide with $a$, $b$ up to a shift of indices. So, by formula~\eqref{monodet}, the determinant of the monodromy of this operator is $s_0^{-n}$. Thus, we have $\lambda = s_0^{-n/3}$, and the result follows.
\end{proof}

We will now show that $X_0 = C$, which implies $s_0 = 1$ and thus proves Theorem \ref{thm1}. To that end, notice that since $s(X_0) = s_0$ is finite and positive, $X_0$ must be located in the open arc $(S_-, S_+)$ (see Figure~\ref{Fig:realpart}). At the same time, since the function $s$ is equal to $1$ at $C$, has a pole at $S_-$, and does not take values $0,1, \infty$ in $(S_-, C)$, it follows that $s > 1$ in $(S_-, C)$. Furthermore, the same argument applied to the function $z$ shows that  $z > 1$ in $(Z_-, C)$, and, in particular, in $(S_-, C)$. But then $X_0$ cannot belong to $(S_-, C)$, because it is not possible that both $s(X_0) = s_0$ and $z(X_0) = s_0^{-n/3}$ are greater than~$1$. Analogously, $s$ and $z$ are both less than $1$ in $(C, S_+)$, so $X_0$ cannot belong there either. Therefore, we must have $X_0 = C$, which implies $s_0 = 1$. But this means that the polygon $P$ is the same as the polygon $P_{sd}$ and hence Poncelet. So, Theorem \ref{thm1} is proved.

\par

\section{Appendix: Duality of difference operators and polygons}\label{sec:app}

    The goal of this appendix is to prove that polygons corresponding to dual difference operators are dual to each other. This seems to be a well-known result, and it explicitly appears as Proposition~4.4.3 in~\cite{morier2014linear}. Here we give a different proof, based on interpretation of difference operators as infinite matrices.
    \begin{proposition}\label{dualdual}
    Let $\D$ be a properly bounded difference operator supported in $[m_-,m_+]$, and let $P = \{v_k\}$ be the corresponding polygon in $\P^{d-1}$, where $d = m_+ - m_-$ is the order of $\D$. Then the dual operator $\D^*$ corresponds to a polygon $P^* = \{v_k^*\}$ in the dual space $(\P^{d-1})^*$ whose $k$'th vertex $v_k^*$ is the hyperplane in $\P^{d-1}$ spanned by the vertices $v_{k + m_-+1}, \dots, v_{k + m_+-1}$ of $P$.
    \end{proposition}
    \begin{proof}
    Let $\D = \sum_{j = m_-}^{m_+} a^j T^j$. Then one can interpret $\D$ as a finite-band matrix \eqref{infMatrix0}
	 whose non-zero diagonals have labels $m_-, \dots, m_+$. (Here and in what follows, the $k$'th diagonal of an infinite matrix is the collection of its entries $a_{ij}$ such that $j - i = k$. In other words, the diagonals are labeled from southwest to northeast, with the main diagonal labeled by $0$.) Note that even though infinite matrices do not form an algebra, any infinite matrix can be multiplied by a finite band matrix. 

\begin{lemma}\label{pseudoinverse}
There exists an infinite matrix $\mathcal L$ such that:\begin{enumerate}\item 
$\D\mathcal L = \mathcal L\D = 0$. \item The diagonals of $\mathcal L$ with labels $-m_++1, \dots, -m_- - 1$ vanish.
\item None of the entries of $\mathcal L$ on the diagonals with labels $-m_+$ and $-m_-$ vanish. \end{enumerate}
\end{lemma}
\begin{remark}
One can think of infinite matrices as formal Laurent series in terms of the shift operator $T$, with coefficients given by sequences. In this language, Lemma \ref{pseudoinverse} states the existence of $\mathcal L$ of the form
 $\sum_{j = -\infty}^{-m_+} b^j T^j +  \sum_{j = -m_-}^{+\infty} b^j T^j,$
 where $b_k^{-m_+} \neq 0$,  $b_k^{-m_-} \neq 0$ for any $k \in \Z$.
\end{remark}
\begin{proof}[Proof of Lemma \ref{pseudoinverse}]
The infinite matrix $\D$ can be regarded as an element of two groups: the group $\GL_{\infty}^+$ of invertible infinite matrices with finitely many non-zero diagonals below the main diagonal, and the group $\GL_{\infty}^-$ of invertible infinite matrices with finitely many non-zero diagonals above the main diagonal. Denote by $\hat \D^{-1}, \check \D^{-1}$ the inverses of $\D$ in these two groups, and set
$
\mathcal L := \hat \D^{-1} - \check \D^{-1}.
$
Then we clearly have $\D\mathcal L = \mathcal L\D = 0$. To see that $\mathcal L$ is of desired form, write $\D$ as 
$
a^{m_-}T^{m_-}(1 + \dots),
$
where the dots denote terms of higher order in $T$. Then the inverse of $(1 + \dots)$ in $\GL^+_{\infty}$ can be computed using the Taylor series $(1+x)^{-1} = 1 - x + \dots$. So, the inverse of $\D$ in $\GL^+_{\infty}$ reads $ \hat \D^{-1}=\,\, (1 + \dots)^{-1}T^{-m_-}(a^{m_-})^{-1}$ and hence is of the form  $\sum_{j = -m_-}^{+\infty} b^j T^j$ with $b_k^{-m_-} \neq 0$ . Likewise, $\check \D^{-1}$ is of the form $ \sum_{j = -\infty}^{-m_+} b^j T^j$ with $b_k^{-m_+} \neq 0$. The result follows.
\end{proof}

We now finish the proof of Proposition \ref{dualdual}. Let $V=\{V_k \in \R^d\}$ be a sequence of lifts of vertices $v_k$ of $P$ such that $\D V = 0$. Then any scalar sequence $\xi \in \Ker \D$ can be obtained from $V$ by means of term-wise application of a linear functional. In particular, since $\D \mathcal L = 0$, this applies to columns of the matrix $\mathcal L$. So, the $j$'th column of $\mathcal L$ is of the form $W_j(V_k)$ for a certain linear functional $W_j \in (\R^d)^*$. Furthermore, since the diagonals of $\mathcal L$ with labels $-m_++1, \dots, -m_- - 1$ vanish, it follows that $W_j$ annihilates $V_{j+m_-+1}, \dots, V_{j + m_+ -1}$. Moreover, since $\mathcal L$ has a non-vanishing diagonal, we have $W_j \neq 0$. Therefore, the projection of $W_j$ to $(\P^{d-1})^* = \P(\R^d)^*$ is exactly the hyperplane spanned by the vertices $v_{j + m_-+1}, \dots, v_{j + m_+-1}$ of $P$. So, to complete the proof, it suffices to show that the sequence of $W_j$'s is annihilated by $\D^*$. To that end, notice that since $\mathcal L\D = 0$, the rows of $\mathcal L$ are annihilated by $\D^*$. But those rows are of the form $W_j(V_k)$, and since $V_k$'s span $\R^d$, it follows that the sequence $W_j$ is annihilated by $\D^*$, as desired. 
\end{proof}
\begin{remark}
 It is also easy to see that the matrix $\mathcal L$ provided by Lemma \ref{pseudoinverse} is unique up to a constant factor. It takes a particularly simple form when the polygon $P$ is closed. To show that, assume for simplicity that $m_- = 0$, so that the operator $\D$ is supported in $[0,d]$. Furthermore, assume that $\D$ is $n$-periodic and has trivial monodromy (in particular, the polygon $P$ corresponding to $\D$ is closed). Then, as shown in \cite{krichever2015commuting}, there exists an $n$-periodic operator $\mathcal R$ supported in $[0,n-d]$ such that $\mathcal R \D = \D \mathcal R = 1 - T^n$ (the operator $\mathcal R$ is closely related to the so-called \textit{Gale dual} of $\D$). Using that, one can find the inverses of $\D$ in $\GL^\pm_{ \infty}$ as 
 $$
\begin{gathered}
 \hat \D^{-1} = \mathcal R (\widehat{1 - T^n})^{-1} = \mathcal R (1 + T^n + T^{2n} + \dots),\\
 \check \D^{-1} = \mathcal R (\widecheck{1 - T^n})^{-1} = -\mathcal RT^{-n}(\widecheck{1 - T^{-n}})^{-1} =  -\mathcal RT^{-n}(1 + T^{-n} + T^{-2n} + \dots) \\ = -\mathcal R(T^{-n} + T^{-2n} + \dots).
\end{gathered}
$$
 As a result, one gets
 $$
 \mathcal L = \hat \D^{-1} - \check \D^{-1} = \mathcal R\sum_{j = -\infty}^{+\infty} T^{jn}.
 $$
\end{remark}

\par\medskip

\bibliographystyle{plain}
\bibliography{pent.bib}

\end{document}